\documentclass[11pt,a4paper]{article}
\pdfoutput=1

\usepackage[margin=1in]{geometry}

\usepackage{fullpage}

\usepackage{color}
\usepackage{graphicx}
\usepackage{amsmath}
\usepackage{amssymb}
\usepackage{amsthm}
\usepackage[noend]{algpseudocode}
\usepackage{algorithm}
\usepackage{subcaption}

\graphicspath{{./figures/}}

\newcounter{theorem}
\newtheorem{lemma}{Lemma}
\newtheorem{theorem}[lemma]{Theorem}

\newtheorem{corollary}[lemma]{Corollary}

\newtheorem{property}[lemma]{Property}

\let\geq\geqslant
\let\leq\leqslant
\newcommand{\IR}{\mathbb{R}}
\newcommand{\IZ}{\mathbb{Z}}
\newcommand{\eps}{\varepsilon}
\DeclareMathOperator{\diam}{diam}
\DeclareMathOperator{\dmin}{dmin}
\DeclareMathOperator{\ball}{b}
\DeclareMathOperator{\dist}{dist}
\DeclareMathOperator{\expansion}{expansion}

\title{Pattern Matching in  Doubling Spaces
\thanks{
This work was supported by Basic Science Research Program through the 
National Research Foundation of Korea (NRF) 
funded by the Ministry of Education (2017R1D1A1B04036529).
}
}

\author{Corentin Allair\\
\'Ecole Polytechnique, Paris, France\\
\texttt{corentin.allair@polytechnique.edu}
\and
Antoine Vigneron\thanks{Corresponding author}\\ 
School of Electrical and Computer Engineering\\ 
UNIST, Republic of Korea\\
\texttt{antoine@unist.ac.kr}
}

\begin{document}
\maketitle

\begin{abstract}
	\normalsize
	We consider the problem of matching a metric space $(X,d_X)$ 
	of size $k$ with a subspace of a metric space $(Y,d_Y)$ of 
	size $n \geq k$, assuming that these two spaces have constant 
	doubling dimension $\delta$.
	More precisely, given an input parameter $\rho \geq 1$, the $\rho$-distortion
	problem is to find 	a one-to-one mapping from $X$ to $Y$ that distorts
	distances by a factor at most $\rho$.
	We first show  by a reduction from $k$-clique
	that, in doubling dimension $\log_2 3$, 
	this problem is NP-hard and W[1]-hard.
	Then we provide a near-linear time approximation
	algorithm for fixed $k$: 
	Given an approximation ratio $0<\eps\leq 1$, and 
	 a 	positive instance 	of the $\rho$-distortion problem, our algorithm returns 
	a solution to	the $(1+\eps)\rho$-distortion problem in time
	$(\rho/\eps)^{O(1)}n \log n$.
	We also show how to extend these results to the minimum distortion problem
	in doubling spaces: We prove the same hardness results, 
	and for fixed $k$, we give a $(1+\eps)$-approximation 
	algorithm running in time $(\dist(X,Y)/\eps)^{O(1)}n^2\log n$, 
	where $\dist(X,Y)$
	denotes the minimum distortion between $X$ and $Y$.
\end{abstract}


\pagenumbering{arabic}

\section{Introduction}

A metric space has {\it doubling dimension} $\delta$ 
if any ball can be covered by at most $2^\delta$  
balls of half its radius.
When $\delta=O(1)$, we say that this space is {\it doubling}.
(See Section~\ref{sec:preliminary}.)
For instance, the Euclidean space $\IR^{d}$ has doubling dimension $O(d)$,
hence doubling spaces are generalizations of fixed-dimensional Euclidean spaces.

In this paper, we study pattern matching problems in doubling spaces. Given two
doubling spaces $(X,d_X)$ and $(Y,d_Y)$  of doubling dimension $\delta$, 
and sizes $|X|=k$ and $|Y|=n$, where $k \leq n$, our goal is to find a subspace 
of $Y$ that resembles the {\it pattern} $X$. More precisely, we consider 
the $\rho$-{\it distortion
problem} and the {\it minimum distortion problem}, which we describe below.

Given $\rho \geq 1$, the $\rho$-distortion problem
is to find,  if it exists, a mapping $\sigma:X \to Y$ such that 
\begin{equation}\label{eq:distortion}
(1/\rho)d_X(x,x') \leq d_Y(\sigma(x),\sigma(x')) \leq \rho d_X(x,x')
\end{equation}
for all $x,x' \in X$. It follows from this definition that  $\sigma$ is injective.

The $\rho$-distortion problem is analogous to the problem of matching
two point-sets in Euclidean space under rigid transformations, which are compositions
of translations and rotations. If, in addition, we allow scaling, then
an analogous problem in general metric spaces is the minimum distortion problem.
The goal is to minimize the distortion 
$\dist(\sigma)=\expansion(\sigma) \times \expansion(\sigma^{-1})$
over all injections $\sigma:X \to Y$, where
\[
	\mathrm{expansion}(\sigma)=
	\max_{\substack{x,x'\in X \\ x \neq x'}} 
	\frac{d_Y(\sigma(x),\sigma(x'))}{d_X(x,x')}
	\quad \text{and} \quad
	\mathrm{expansion}(\sigma^{-1})=
		\max_{\substack{x,x'\in X \\ x \neq x'}} 
	\frac{d_X(x,x')}{d_Y(\sigma(x),\sigma(x'))}.
\]
The minimum of $\dist(\sigma)$ over all injections $\sigma: X \to Y$ 
is denoted $\dist(X,Y)$, and it is easy to see that $\dist(X,Y) \geq 1$.
The minimum distortion problem was introduced by Kenyon et al.~\cite{Kenyon09}
in the case where $k=n$, and thus $\sigma$ is a bijection.

Motivated by applications to natural language processing, bioinformatics
and computer vision, Ding and Ye~\cite{DingY19} recently proposed
a practical algorithm for a pattern
matching problem in doubling spaces.
However, this algorithm may only return an approximation of a local
minimum, and its
time bound is not given as a function of the input size.
One of our goals is thus to provide a provably efficient algorithm
for pattern matching in doubling spaces. Another motivation for our work is that,
even though the complexity of the minimum distortion problem has been studied 
for several types of metrics, it appears that no result was previously
known for two doubling metrics. (See the comparison with previous
work below.)

\paragraph{Our results.}

We first give a hardness result: We show that for any $\rho \geq 1$,
the  $k$-clique problem reduces
to $\rho$-distortion in doubling dimension $\log_2 3$. It implies that
the $\rho$-distortion problem is NP-hard, and is W[1]-hard when
parameterized by $k$  (Corollary~\ref{cor:hardness}). 
It also shows that this problem cannot
be solved in time $f(k) \cdot n^{o(k)}$ for any computable function $f$, 
unless the exponential time hypothesis (ETH) is false
(Corollary~\ref{cor:ETH}).

On the positive side, we present a near-linear time approximation algorithm
for small values of $k$. More precisely, if $\rho \geq 1$ 
and  $0<\eps\leq 1$, our algorithm returns in 
$2^{O(k^2 \log k)} (\rho^2/\eps)^{2k\delta}n\log n$  time
a solution to the $(1+\eps)\rho$-distortion problem whenever
a solution to the $\rho$-distortion problem exists
(Theorem~\ref{th:nospread}). In this time bound, it is reasonable to assume that
$\rho$ is a small constant, say $\rho \leq 10$, as a larger value 
would mean that we allow
a relative error of more than 900\% in the quality of the matching, which is 
probably too much for most applications.

We also show how to extend these results to the minimum distortion problem.
In particular, we show that the minimum distortion problem cannot be solved
in time $f(k) \cdot n^{o(k)}$ for any computable function $f$, 
unless ETH is false, and we give a
$(1+\eps)$-approximation algorithm running in  
$2^{O(k^2 \log k)} (\dist(X,Y)^{2k\delta}/\eps^{2k\delta+O(1)})n^2\log n$
time (Theorem~\ref{thm:optimization}).
Here again, it is reasonable to assume that $\dist(X,Y)=O(1)$, 
and then for any fixed $k$, this algorithm is an FPTAS with running time
$(1/\eps)^{O(1)} n^2 \log n$.

\paragraph{Comparison with previous work.}

One of the main differences between our results and previous work on
point pattern matching under rigid transformations, or on the minimum distortion problem,
is that we parameterize the problem by $k$, and hence
$k$ is regarded as a small number. The advantage is that the dependency
of our time bounds in $n$ are low (near-linear or near-quadratic). However,
we obtain an exponential dependency in $k$, which may be unavoidable due to
our hardness results.

Geometric point pattern matching problems have been studied extensively.
(See for instance the survey by Alt and Guibas~\cite{Alt00}.) In the
fixed-dimensional Euclidean space $\IR^d$, these problems are usually tractable, 
as the space of transformations has a constant number of degrees of freedom.
For instance, when $k=n$,
we may want to decide whether $X$ and $Y$ are congruent, which means that
there is a rigid transformation $\mu$ such that $\mu(X)=Y$. Alt et al.~\cite{Alt88} 
showed how
to find such a transformation in time $O(n^{d-2}\log n)$, when it exists.
In practice, however, we cannot expect that point coordinates are known exactly,
so it is unlikely that an exact match exists.  We may thus want to find the smallest
$\eps>0$ such that each point of $X$ is brought to distance at most $\eps$
from a point in $Y$. (In other words, we allow an additive error $\eps$.)
Chew et al.  gave an  $O(k^3n^2\log^2(nk))$-time algorithm
to solve this problem in the plane under rigid transformations~\cite{Chew97}.

As mentioned above, to the best of our knowledge, the only work published so far on pattern
matching in doubling spaces presents  a practical algorithm for matching
two doubling spaces~\cite{DingY19}. However it has not been proven to return a good 
approximation of the optimal solution in the worst case. Other problems studied
in doubling spaces include
approximate near-neighbor searching~\cite{AryaMVX08,CG06,HM06},
spanners~\cite{Borradaile19}, routing~\cite{ChanLNS15,GottliebR08}, 
TSP~\cite{Talwar04}, clustering~\cite{Friggstad19}, 
Steiner forest~\cite{Chan18} \dots

The minimum distortion problem has been studied under various metrics, when $k=n$.
Kenyon et al.~\cite{Kenyon09} gave a polynomial-time algorithm for
line metrics (1-dimensional point sets) when $\dist(X,Y)<5+2\sqrt 6$. 
They also gave an algorithm that computes $\dist(X,Y)$ when $d_X$ is
the metric associated with an unweighted graph over $X$ and $d_Y$ is the
metric associated with a bounded degree tree over $Y$,
with a running time  exponential in the maximum degree and doubly
exponential in $\dist(X,Y)$. 
For general metrics, the minimum distortion
is hard to approximate within a factor less than $\log^{1/4-\gamma} n$, for
any $\gamma>0$~\cite{Khot07}. 
Hall and Papadimitriou~\cite{Hall05} showed that even for line metrics,
the distortion is hard to approximate when it is large.

When $k \leq n$, {F}ellows et al.~\cite{Fellows05} showed that the problem of deciding whether
$\dist(X,Y) \leq D$ is fixed-parameter tractable when parameterized
by $\Delta$ and $D$, where $d_X$ is the metric associated with
an unweighted graph over $X$, and $d_Y$ is the metric associated with
a tree of degree at most $\Delta$ over $Y$. For two unweighted graph metrics,
Cygan et al.~\cite{Cygan17} showed that the problem cannot be solved in
time $2^{o(n \log n)}$ unless ETH is false. When $X$ is an arbitrary finite
metric space and $Y$ is a subset of the real line,  Nayyeri and Raichel~\cite{Nayyeri15}
showed that a constant-factor approximation of $\dist(X,Y)$ can be computed
in time $\Phi(X)^{O((\dist(x,y))^2)}(kn)^{O(1)}$, where $\Phi(X)$ is the
spread of $X$. (See Section~\ref{sec:preliminary}.)

In summary, the previously known theoretical results on the minimum distortion 
problem are either
hardness results, or algorithms for cases where $d_Y$ is a subset of a line metric
or a tree metric. The algorithms presented in this paper, on the other hand,
apply when $d_X$ and $d_Y$ are doubling metrics, which are generalizations
of fixed-dimensional Euclidean metrics.

\paragraph{Our approach.}

In Section~\ref{sec:hardness}, we present hardness results on the
$\rho$-distortion problem. We reduce an instance $G(V,E)$ of $k$-clique 
to an instance of $\rho$-distortion consisting of
two metric spaces $(X,d_X)$ and $(Y,d_Y)$ of sizes $k$ and $km$, respectively,
where $m=|V|$.
The pattern $(X,d_X)$ is an ultrametric, with exponentially increasing
distances.
The space $(Y,d_Y)$ consists of $k$ rings, each ring consisting
of $m$ points regularly spaced on a circle of perimeter 1.
Each of these $m$ points is associated with a vertex of the $k$-clique instance.
(See \figurename~\ref{fig:ring}a.) The distances between the
rings increase exponentially (\figurename~\ref{fig:rings}), 
and the input graph is encoded
by having slightly longer edges for pairs of vertices lying in
different rings that correspond to edges in the input graph
(\figurename~\ref{fig:ring}b.) We prove that these two
spaces have doubling dimension $\log_2 3$, and that this instance
of $\rho$-distortion is equivalent to the $k$-clique instance
we started from.

In Section~\ref{sec:decision}, we give a self-contained description of
a first approximation algorithm for the $\rho$-distortion that runs in time  
$2^{O(k^2 \log k)} (\rho^2/\eps)^{3k\delta}n+O(kn\log \Phi(Y))$
where $\Phi(Y)$ is the {\it spread} of $Y$. (See Section~\ref{sec:preliminary}.)
We first construct a  {\it navigating net} over $Y$~\cite{KL04}. 
A navigating net is essentially a
coordinate-free quadtree that records a metric space. 
It represents $Y$ at all resolutions $r$ where $r$ is a powers of 2.
At each scale $r$, our navigating net records an {\it $r$-net} $Y_r$ of $Y$,
which is a subset of $Y$ whose points are a distance
at least $r$ apart, and such that the radius-$r$ balls centered
at $Y_r$ cover $Y$. (See \figurename~\ref{fig:nnet}.)

Let $r_X$ be the smallest scale that is at least $\rho$ times the diameter of $X$.
Our algorithm constructs, for each point $y \in Y_{r_X}$,
a sparse set of matchings 
whose images are in the radius-$3r_X$ ball centered at $y$.
(By sparse, we mean that any two such matchings send at least one
point of $X$ to two points of $Y$ that are a distance at least
$\eps r_X/(2\rho^2)$ apart.)
The union of these sets of matchings 
over all $y \in Y_{r_X}$ is denoted $L(X,\eps,r_X)$,
and we show that any solution to the $\rho$-distortion problem
is close to at least one matching in $L(X,\eps,r_X)$.

We compute $L(X,\eps,r_X)$ recursively from  $L(P,\beta,r_P)$
and $L(Q,\beta,r_Q)$ for $\beta=\eps/(8k-8)$, where $P$
and $Q$ form a partition of $X$. More precisely,
we obtain $P$ and $Q$ by running Kruskal's algorithm on $X$,
and stopping at the second last step. It 
ensure that $P$ and $Q$ are well separated, and it follows that  
any $\rho$-matching
$\bar \sigma:X \to Y$ can be approximated by a combination of two
$(1+\beta)\rho$-matchings $\bar \sigma^P_\beta:P \to Y$ and
$\bar \sigma^Q_\beta:Q \to Y$ recorded  in $L(P,\beta,r_P)$
and $L(Q,\beta,r_Q)$, respectively.
(See \figurename~\ref{fig:split}.)
After computing $L(X,\eps,r_X)$, we simply return
one of the matchings that it records,  if any.

In Section~\ref{sec:improved}, we show how to improve the time bound to
$2^{O(k^2 \log k)} (\rho^2/\eps)^{2k\delta} n\log n$.
We achieve it using the approximate near-neighbor (ANN) data structure
by Cole and Gottlieb~\cite{CG06}. First, this data structure allows
us to efficiently prune the sets of matchings recorded at layer $r_X$,
by only inserting a new matching if it is far enough from all previously
inserted matching. It can be checked by performing a constant number of 
ANN queries in the set of matchings. As the space of matchings is doubling
(Corollary~\ref{cor:dimmatching}), it takes logarithmic time.
This saves a factor $k(\rho^2/\eps)^{k\delta}/\log n$ in our time
bound.
Second, instead of computing the whole navigating net, which takes
time $O(n \log \Phi(Y))$, we show how to compute
any layer in $O(n \log n)$ time using ANN queries. As our algorithm
only requires $k$ layers of the navigating net, it removes the
dependency on $\Phi(Y)$ from the time bound.

Finally, in Section~\ref{sec:mindist},
we show how to extend our results on the $\rho$-distortion
problem to the minimum distortion problem. 
For the hardness result,
it suffices to add an extra point to $X$ and $Y$ that is far enough
from the other points, in order to make the reduction work. For the
algorithms, we use the reduction by Kenyon et al.~\cite{Kenyon09}
of the minimum distortion problem to the $\rho$-distortion problem,
which we speed-up using exponential search, and using a well-separated
pairs decomposition, which allows us to reduce the number of candidate values
for $\dist(X,Y)$.

\section{Notation and Preliminary}\label{sec:preliminary}

Let $(S,d_S)$ be a finite metric space.
The ball $\ball(x,r)$ centered 
at $x$ with radius $r$ 
is the set of points $x' \in S$ such that $d_S(x,x') \leq r$.
The minimum and maximum interpoint distances 
in $S$ are denoted 
$\dmin(S)$ and $\diam(S)$, respectively.
In other words, $\diam(S)$ is the diameter of $S$.
The {\em spread} of $S$ is
the ratio $\Phi(S)=\diam(S)/\dmin(S)$.
The distance from a point $x$ to a subset $T$ of $S$
is $d_S(x,T)=\min_{t \in T}d_S(x,t)$. The distance between two
sets $T$ and $U$ is 
$d_S(T,U)=\min_{t \in T, u \in U}d_S(t,u)$.


A metric space $(S,d_S)$ has {\em doubling dimension}
$\delta$ if any ball of radius $r$ is contained in the union of at
most $2^\delta$ balls of radius $r/2$. When $\delta=O(1)$, we say
that this space is {\em doubling}. 
This notion of dimension generalizes the dimension of a Euclidean space:
In particular, the Euclidean space $\IR^d$
has doubling dimension $O(d)$~\cite{GKL03}.
In this paper, we will consider spaces
of constant doubling dimension, so we assume that $\delta=O(1)$.
We will need the following packing lemma:
\begin{lemma}[\cite{KN19}]\label{lem:packing}
	If a metric space $(S,d_S)$ has doubling dimension $\delta$, 
	then $|S|\leq (4\Phi(S))^\delta$.
\end{lemma}
We will also make use of the fact that a product of
doubling metrics is  doubling. It was probably known,
but we could not find a reference, so we include a proof.
\begin{lemma}\label{lem:productmetric}
	Let $(S_1,d_1),\dots,(S_k,d_k)$ be metric spaces with
	doubling dimensions $\delta_1,\dots,\delta_k$, respectively. 
	Then the product
	metric $(S,d_S)$ where $S=S_1\times \dots \times S_k$ and
	\[	d_S((u_1,\dots,u_k),(v_1,\dots,v_k))
		=\max(d_1(u_1,v_1),\dots,d_k(u_k,v_k))
	\] for all $(u_1,\dots,u_k)$, $(v_1,\dots,v_k) \in S$ 
	has doubling dimension $\delta_1+\dots+\delta_k$.
\end{lemma}
\begin{proof}
	We prove this lemma for $k=2$; the general case follows by induction
	on $k$.
	Let $b=\ball((u_1,u_2),r)$ be a ball of radius $r$ in $S$.
	Then $b=b_1 \times b_2$ where $b_1=\ball(u_1,r)$ and $b_2=\ball(u_2,r)$
	are balls of radius $r$ in $S_1$ and $S_2$, respectively.
	So $b_1$ is covered by $m_1 \leq 2^{\delta_1}$ balls 
	$b_1^1,\dots,b_1^{m_1}$ balls of radius  $r/2$. Similarly,
	$b_2$ is covered by $m_2 \leq 2^{\delta_2}$ balls 
	$b_2^1,\dots,b_2^{m_2}$ balls of radius  $r/2$. 

	Let $(v_1,v_2) \in b$.
	Then we have $d_1(u_1,v_1) \leq r$ and $d_2(u_2,v_2) \leq r$, which means
	that $v_1 \in b_1$ and $v_2 \in b_2$. So there exist
	$i_1$ and $i_2$ such that $v_1 \in b_1^{i_1}$ and $v_2 \in  b_2^{i_2}$.
	In other words, $(v_1,v_2) \in  b_1^{i_1} \times  b_2^{i_2}$.
	We have just proved that $b$ is contained in the union of the 
	Cartesian products 	$b_1^{i_1} \times b_2^{i_2}$. There are 
	$m_1 \times m_2 \leq 2^{\delta_1+\delta_2}$ such Cartesian
	products, and each one of them  is a ball of radius $r/2$ in $S$. 
	Therefore, $(S,d)$ 	has doubling dimension $\delta_1+\delta_2$. 
\end{proof}

We call a mapping $\sigma$ satisfying Equation~\eqref{eq:distortion} a 
{\em $\rho$-matching} from $X$ to $Y$. The distance between two 
matchings $\sigma$ and $\sigma'$ from $X$ to $Y$ 
is $d_M(\sigma,\sigma')=\max_{x \in X}d_Y(\sigma(x),\sigma'(x))$.
We denote by $x_1,x_2,\dots,x_k$ the $k$ elements of $X$.
So a matching $\sigma:X \to Y$ can be identified with a sequence
of $k$ points $(y_1,\dots,y_k)$ where $y_1=\sigma(x_1),\dots,y_k=\sigma(x_k)$.
In other words, the space of matchings from $X$ to $Y$ can be identified
with $(Y^k,d_M)$. Then it follows from Lemma~\ref{lem:productmetric} that:
\begin{corollary}\label{cor:dimmatching}
	The space of matchings from $X$ to $Y$ has doubling dimension $k\delta$.
\end{corollary}

\section{Reduction from $k$-clique}\label{sec:hardness}

Given an integer $k\geq 1$ and a graph $G(V,E)$, 
the $k$-{\it clique} problem is to 
decide whether there exists a subset $C \subseteq V$ of $k$ vertices such that
any two of these vertices are connected by an edge in $E$. This subset $C$
is called a $k$-clique.
In this section, we present a reduction from the $k$-clique problem to 
the $\rho$-distortion problem. 

\subsection{Construction}
\label{sec:construction}

So let $G(V,E)$ be an instance of $k$-clique with $m$ vertices.
We denote $V=\{v_1,\dots,v_m\}$, and
we assume that $m\geq 24$. For any $\rho \geq 1$, we will show how to 
construct
an equivalent instance of the $\rho$-distortion problem consisting of
two metric spaces $(X,d_X)$ and $(Y,d_Y)$ of respective sizes $k$ and $km$,
and of doubling dimension $\log_2 3$.

\begin{figure}[h]
	\begin{subfigure}[t]{0.39\textwidth}
		\centering
		\includegraphics{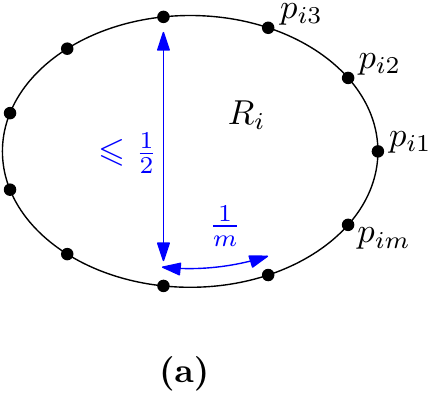}
	\end{subfigure}
	\begin{subfigure}[t]{0.59\textwidth}
		\centering
		\includegraphics{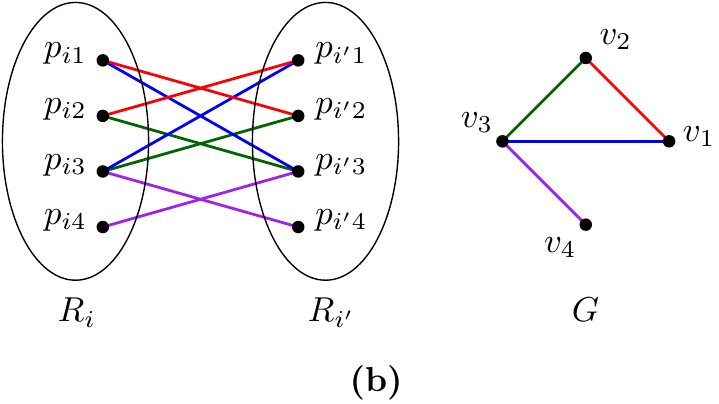}
	\end{subfigure}
	\caption{{\bf(a)} An $m$-point ring gadget $R_i$. 
		{\bf(b)} Two rings $R_i$ and $R_{i'}$.
			The edges drawn between $R_i$ and $R_{i'}$ represent distances
			equal to ${2}^{\max(i,i')}$, and the other distances between a point
			in $R_i$ and a point in $R_{i'}$ are ${2}^{\max(i,i')}-(1/m)$
		\label{fig:ring}}
\end{figure}

Let us first build $Y$, and its associated metric $d_Y$.
We define an $m$-point {\it ring gadget} (\figurename~\ref{fig:ring}a)
as a set $R_i=\{p_{i1},..,p_{im}\}$  of $m$ points
spaced regularly on a circle of perimeter 1, so that the distance between 
any two points is the usual distance along this circle:
\begin{equation}
	d_Y(p_{ij},p_{ij'}) =(1/m) \cdot \min{(j'-j, m+j-j')} 
	\text{ whenever } 1\leq j \leq j' \leq m.
	\label{eq:onering}
\end{equation}

\begin{figure}[h]
	\centering
	\includegraphics{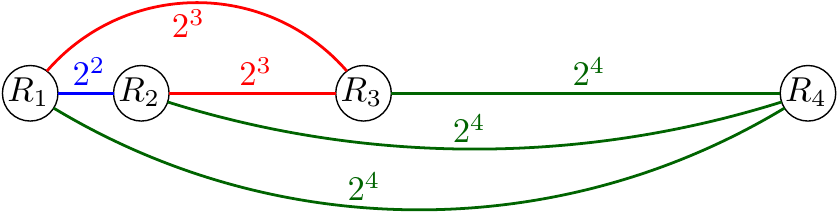}
	\caption{Distances between rings. For instance, the distance
		between any point in $R_2$ and any point in $R_4$ is ${2}^4$ or
		${2}^4-(1/m)$.
	\label{fig:rings}}
\end{figure}

We define $Y=R_1\cup \dots \cup R_k$ as the disjoint union of $k$ ring gadgets. 
The index $j$ of the point $p_{ij} \in R_i$ 
corresponds to the vertex  $v_j$ of $G$. 
The distances between points in different rings 
are defined as follows. 
(See \figurename~\ref{fig:ring}b and \figurename~\ref{fig:rings}.)
For any $i,i' \in \{1,\dots,k\}$ such that $i \neq i'$, and for any 
$j,j' \in \{1,\dots,m\}$, 
\begin{equation}
	d_Y(p_{ij},p_{i'j'}) = 
		\begin{cases}
			 {2}^{\max(i,i')} & \text{if } (v_j,v_{j'}) \in E, \text{ and}	\\
			 {2}^{\max(i,i')}-(1/m) & \text{otherwise}.
		\end{cases}
	\label{eq:distrings}
\end{equation}
Thus, the distance between two vertices in different rings is a power of ${2}$ 
if and only if their associated vertices in $G$ are connected by an edge. The
distance between two points in the same ring is given by Equation~\eqref{eq:onering}.

The pattern set $X$ consists of $k$ distinct points $x_1,\dots,x_k$.
We associate with $X$ the following distance function $d_X$:
\[
	d_X(x_i,x_{i'})=\begin{cases}
		0 & \text{if  $i=i'$, and} \\
		 {2}^{\max(i,i')}\rho  & \text{otherwise.}
	\end{cases}
\]

We now prove that the two spaces $(X,d_X)$ and $(Y,d_Y)$ are metric spaces with doubling
dimension $\log_2 3$. 

\begin{lemma}\label{lem:Ymetric}
The distance function $d_Y:Y^2 \to \IR$  defined above is a metric.
\end{lemma}

\begin{proof}
First observe that $d_Y$ is non-negative, symmetric, and 
that $d_Y(p,p')=0$ if and only if $p=p'$.
Within one ring, $d_Y$ respects the triangle inequality because it
is the usual distance along a circle. So we only need to prove that
the triangle inequality holds for points lying in two or three different rings.
Let  $i,i',i''\in \{1,\dots,k\}$ and $j,j',j'' \in \{1,\dots,m\}$.  
If $i$, $i'$ and $i''$ are distinct,  we have
\[
	d_Y(p_{ij},p_{i'j'})+d_Y(p_{i'j'},p_{i''j''}) 
		 \geq {2}^{\max(i,i')} + {2}^{\max(i',i'')} - \frac 2 m 
		 \geq {2}^{\max(i,i'')} \geq d_Y(p_{ij},p_{i''j''}).
\]
If $i \neq i'$ and $j' \neq j''$, then we have
\[
	d_Y(p_{ij},p_{i'j'})+d_Y(p_{i'j'},p_{i'j''}) 
	 \geq \left({2}^{\max(i,i')} - \frac 1 m \right)+ \frac 1 m 
	 \geq {2}^{\max(i,i')} \geq d_Y(p_{ij},p_{i'j''}).
\]
If $i \neq i'$, then we have
\[
	d_Y(p_{ij},p_{i'j'})+d_Y(p_{i'j'},p_{ij''}) 
	 \geq 2\left({2}^{\max(i,i')} - \frac 1 m \right)
	 \geq {2}^i \geq d_Y(p_{ij},p_{ij''}).
\]
This completes the proof that $d_Y$ is a metric.
\end{proof}

\begin{lemma}\label{lem:Ydoubling}
The metric space $(Y,d_Y)$ has doubling dimension $\log_2 3$.
\end{lemma}
\begin{proof}
Let $b_0 =\ball(p,r)$ be an arbitrary ball, where $p=p_{ij}$. 
We assume that the index $j'$ in $p_{ij'}$ is defined modulo $m$,
hence $p_{ij'}=p_{i(j'+m)}$. Let $t = \lfloor \log_{2} (r+1/m) \rfloor$.

Suppose that  $t<i$. Then $\log_{2} (r+1/m) < i$, and thus $r < {2}^i-1/m$.
As the distance from any point in $R_i$ to any point in another ring is at least
${2}^i-1/m$, it follows that  $b_0 \subseteq R_i$. 
\begin{itemize} 
	\item If $\lceil mr/2 \rceil \geq m/4$, then $1+mr/2\geq m/4$.
		As $m \geq 24$, it implies that $r \geq 5/12$. 
		Let $c=p_{i1}$, $c'=p_{i(1+\lfloor m/3 \rfloor)}$ and
		$c''=p_{i(1-\lfloor m/3 \rfloor)}$.
		As $m \geq 24$, we have $7m/24 \leq \lfloor m/3 \rfloor \leq m/3$,
		and thus $d_Y(c,c')=d_Y(c,c'') \leq 1/3$ and $d_Y(c',c'') \leq 1-2\cdot(7/24)= 5/12$.
		So we can cover the whole ring $R_i$ (and thus $b_0$) 
		with the balls of radius $r/2$ centered at $c$, $c'$ and $c''$.
	\item Otherwise, the points 
		$p'=p_{i(j+\lceil mr/2 \rceil)}$ and $p''=p_{i(j-\lceil mr/2 \rceil)}$ 
		lie in the open half-circle centered at $p$,
		that is, $d_Y(p,p') < 1/4$ and $d_Y(p,p'') < 1/4$.
		As $mr/2 \leq \lceil mr/2 \rceil < mr/2+1$, it implies that 
		$r/2 \leq d_Y(p,p') < r/2+1/m$ and $r/2 \leq d_Y(p,p'') < r/2+1/m$. 
		Therefore $\ball(p',r/2) \cup \ball(p'',r/2)$
		contains $b_0 \setminus \{p\}$, and thus 
		$b_0 \subseteq \ball(p,r/2) \cup \ball(p',r/2)\cup \ball(p'',r/2)$.
\end{itemize}

Now suppose that $t \geq i$ and $t \geq 3$. As $t \leq \log_{2} (r+1/m)< t+1$, we have 
${2}^t-1/m \leq r < {2}^{t+1}-1/m$.
For any $i' \geq t+1$, the distance from $p$ to any point in $R_{i'}$ is at least 
${2}^{t+1}-1/m$, so $b_0 \subseteq R_1 \cup \dots \cup R_t$.
We consider three balls $b_1$, $b_{t-1}$ and $b_t$ of radius $r/2$  
centered at an arbitrary point  $c_1 \in R_1$, $c_{t-1} \in R_{t-1}$ 
and $c_t \in R_t$, respectively.
Since 
$
	\frac r 2 > \frac{{2}^t}2-\frac 1{2m} 
	\geq \frac 2 2-\frac 1{48} = \frac{47}{48}, 
$
and  $\diam(R_t)\leq 1/2$, we have $R_t \subseteq b_t$.
Similarly, we have $R_{t-1} \subseteq b_{t-1}$. 
Let  $q \in R_1 \cup \dots \cup R_{t-2}$. 
Then we have $d_Y(c_1,q) \leq 2^{t-2}$.
As $r/2 \geq 2^{t-1}-1/(2m) \geq 2^{t-2}$, it follows that $q \in b_1$. So we just
proved that $R_1 \cup \dots \cup R_{t-2} \subseteq b_1$.
Therefore, 
$b_0 \subseteq R_1 \cup \dots \cup R_t \subseteq b_1\cup b_{t-1} \cup b_t$.

Similarly, if $t \geq i$ and $t=1$, then $t=1=i$ and we have $b_0 \subseteq b_1$.
If $t \geq i$ and $t=2$, then $b_0 \subseteq b_1 \cup b_2$.

In any case, $b_0$ is contained in the union of at most 3 balls of half its radius, and thus $(Y, d_y)$ has doubling dimension $\log_2 3$.
\end{proof}
A simpler version of the arguments in the proofs of Lemma~\ref{lem:Ymetric}
and~\ref{lem:Ydoubling} yields the following.
\begin{lemma}\label{lem:Xdoubling}
	$(X,d_X)$ is a metric space of doubling dimension $\log_2 3$.
\end{lemma}

\subsection{Proof of Correctness}
We now prove that our reduction of $k$-clique to the $\rho$-distortion problem 
is correct.
So given an instance $G=(V,E)$ of the $k$-clique problem,
we construct the metric spaces $(X,d_X)$ and $(Y,d_Y)$ as described
above. These two metric spaces form an instance of the $\rho${-distortion problem}.
We need to show that these two instances of $k$-clique and $\rho$-distortion are
equivalent.

We first assume that $G$ is a positive instance of $k$-clique.
So there is a clique $\{v_{j_1},\dots,v_{j_k}\}$ of size $k$ in $G$.
Let $\sigma:X \to Y$ be the matching defined by $\sigma(x_i)=p_{ij_i}$
for all $i \in \{1,\dots,k\}$. For any $i \neq i'$, there is an edge
between $v_{j_i}$ and $v_{j_{i'}}$ in $G$. Therefore, we have 
$
d_Y( p_{ij_i}, p_{i'j_{i'}}) = {2}^{\max(i,i')} 
= \frac{1}{\rho} d_X(x_i,x_{i'}).
$
It means that 
$d_Y(\sigma(x_i),\sigma(x_{i'}))=\frac{1}{\rho} d_X(x_i,x_{i'})$, 
and thus
$\sigma$ is a solution to our instance of the $\rho$-distortion problem.

We now prove the converse. Let $\sigma:X \to Y$ be a solution to 
our instance of the $\rho$-distortion problem. We denote $y_i=\sigma(x_i)$ for
each $x_i \in X$. Each $y_i$ corresponds to a vertex $w_i \in V$. More precisely,
we have $y_i=p_{i'j'}$ for some indices $i',j'$, and we set $w_i=v_{j'}$.
We want to show that $w_1,\dots,w_k$ is a $k$-clique in $G$. 

\begin{lemma}
\label{lem:natural_gadget}
Each ring contains exactly one point $y_i$, and $y_i \in R_i$ whenever $i \geq 3$.
\end{lemma}
\begin{proof}
Let $1 \leq i_1 < i_2 \leq k$. Then $d_X(x_{i_1},x_{i_2})={2}^{i_2}\rho$. As $\sigma$
is a solution to the $\rho$-distortion problem, it implies that 
$
	d_Y(y_{i_1},y_{i_2})=d_Y(\sigma(x_{i_1}),\sigma(x_{i_2}))\geq {2}^{i_2} \geq 4.
$
As the diameter of each ring is at most $1/2$, it follows that $y_{i_1}$ and 
$y_{i_2}$
lie in different rings. Hence each ring contains exactly one point $y_i$.

We now prove the second part of the lemma, so we may assume that $k \geq 3$.
The point $x_k$ is at distance ${2}^{k}\rho$ from each 
point $x_1,\dots,x_{k-1}$. As $\sigma$ is a solution to the $\rho$-distortion problem,
the distance from $y_k$ to any point $y_1,\dots,y_{k-1}$ is at least 
${2}^k$, and since it is the largest interpoint distance in $Y$, it implies
that $d_Y(y_i,y_k)={2}^k$ for any $1\leq i < k$.
This is only possible if $y_k \in R_k$.

Applying the same argument repeatedly to $y_{k-1},\dots,y_3$ shows that $y_i \in R_i$
for any $i \in \{3,\dots,k-1\}$. (As $d_X(x_1,x_i)=d_X(x_2,x_i)$ for any $i \geq 3$,
the points $x_1$ and $x_2$ are indistinguishable, hence we have either $y_1 \in R_1$ and
$y_2 \in R_2$, or $y_1 \in R_2$ and $y_2 \in R_1$.)
\end{proof}

We can now prove that:
\begin{lemma}
The vertices $\{w_1,..,w_k\}$ form a $k$-clique in $G$.
\end{lemma}
\begin{proof}
Let $1 \leq i' < i \leq k$. By Equation~\eqref{eq:distrings},
the distance $d_Y(y_i,y_{i'})$ can take only two values: it
is equal to ${2}^i$ if $(w_i,w_{i'}) \in E$, and ${2}^i-(1/m)$
otherwise. As $\sigma$ is a solution to the $\rho$-distortion problem,
we also have
$
	d_Y(y_i,y_{i'}) = d_Y(\sigma(x_i),\sigma(x_{i'}))
		\geq \frac 1 {\rho}d_X(x_i,x_{i'})= {2}^i.
$
It implies
that $d_Y(y_i,y_{i'})={2}^i$, and that $(w_{i},w_{i'}) \in E$.
\end{proof}

In summary, we obtained the following result.

\begin{theorem}\label{thm:reduction}
	The graph $G$ admits a $k$-clique if and only if 
	$(X,d_X)$, $(Y,d_Y)$ is a positive instance of the $\rho$-distortion problem.
\end{theorem}

\subsection{Consequences}

The construction of $(X,d_X)$ and $(Y,d_Y)$ from $G$ is performed in polynomial time.
It is also an FPT-reduction with parameter $k$. As $k$-clique is NP-complete
and $W[1]$-hard, it follows that:
\begin{corollary}
	\label{cor:hardness}
	The $\rho$-distortion problem for doubling spaces of dimension $\log_2 3$ 
	is NP-hard, and is $W[1]$-hard when  parameterized
	by $k$.
\end{corollary}

Unless the exponential time hypothesis (ETH) is false, our reduction also shows that
that the $\rho$-distortion problem cannot be solved in time $n^{o(k)}$. 
More precisely, it follows from a known hardness
result on $k$-clique~\cite[Theorem 14.21]{Cygan15} that: 
\begin{corollary}\label{cor:ETH}
	The $\rho${-distortion problem} for doubling spaces of dimension $\log_2 3$ 
  cannot be solved in time $f(k) \cdot n^{o(k)}$ for any computable function
	$f$, unless  ETH is false.
\end{corollary}

\section{Approximation algorithm for the $\rho$-distortion problem}\label{sec:decision}

In this section, we present an approximation algorithm
for the $\rho$-distortion problem. 
We assume that the doubling dimension is constant, that is, $\delta=O(1)$.
As we saw in Section~\ref{sec:hardness}, the $\rho$-distortion problem
is hard, so we relax the problem slightly: Given a 
parameter $0<\eps\leq 1$, the $(\rho,\eps)$-{\it distortion 
problem} is to find a $(1+\eps)\rho$-matching
whenever a $\rho$-matching exists. 
If there is no $\rho$-matching,
then our algorithm either returns a $(1+\eps)\rho$-matching,
or it does not return any result.

\subsection{Navigating nets}\label{sec:nnet}
Our algorithm  records $Y$ in a {\em navigating net}, which is a data
structure representing $Y$ at different resolutions. 
(See \figurename~\ref{fig:nnet}.)
This structure was
introduced by Krauthgamer and Lee~\cite{KL04}. We will use a slightly
modified version of it. Our version of the navigating net
has the advantages that each layer is an $r$-net of $Y$ (see definition
below),
and that it can easily be computed in $O(n \log \Phi(Y))$ time. On the
other hand, it does not allow efficient deletions.
Several other variations exists. 
In particular, Cole and Gottlieb~\cite{CG06} and Har-Peled and Mendel~\cite{HM06}
presented more involved data structures whose time bounds do not
depend on the spread.

\begin{figure}[h]
	\centering
	\includegraphics{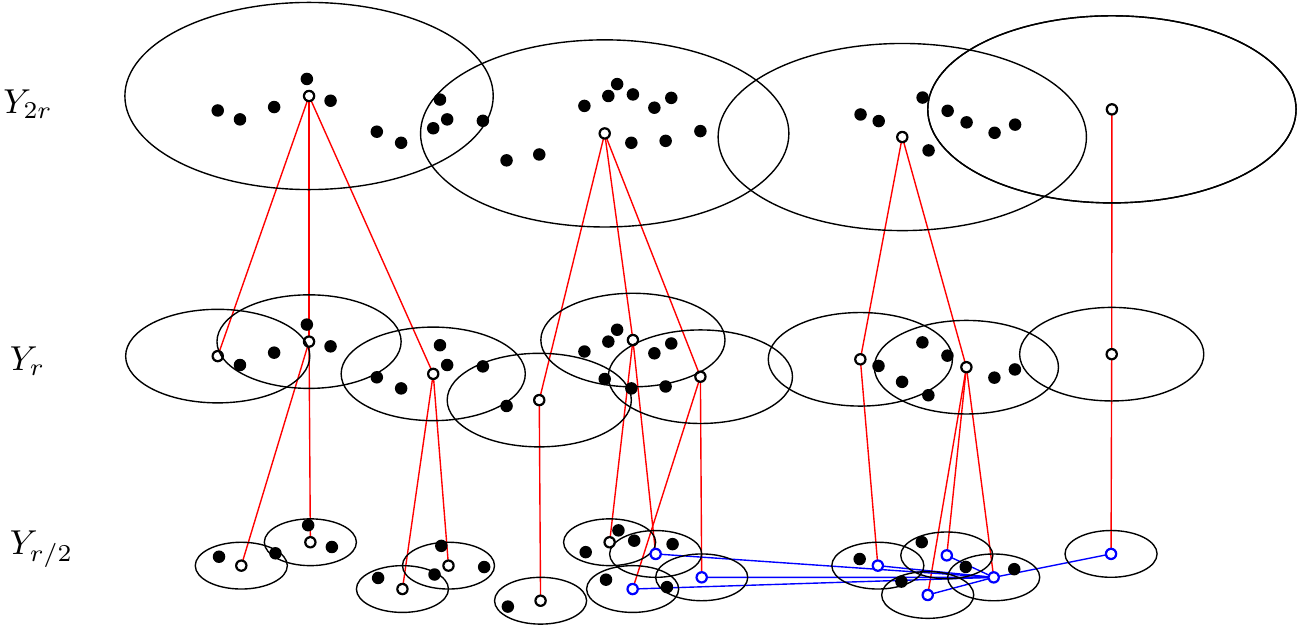}
	\caption{A navigating net at three consecutive scales $r/2$, $r$ and $2r$.
	The white dots represent points of the $r$-nets $Y_{r/2}$, $Y_r$ and
	$Y_{2r}$. The ellipses represent the balls of radius $r/2$, $r$ and $2r$
	centered at the points in $Y_{r/2}$, $Y_r$ and 	$Y_{2r}$, respectively.
	The blue edges are the horizontal edges incident to one of the vertices
	of $Y_{r/2}$. The red edges are the vertical edges.
	\label{fig:nnet}}
\end{figure}

For any $r \geq 0$, an {\em $r$-net} of $Y$ is a subset 
$Y_r \subseteq Y$ such that $\dmin(Y_r) \geq r$ and 
$Y \subseteq \bigcup_{y \in Y_r} \ball(y,r)$. An $r$-net can be 
constructed incrementally by repeatedly adding new points that
lie outside of the current union of balls, until $Y$ is
completely covered. Intuitively, an $r$-net  represents
$Y$ at resolution $r$.

A {\em scale} is a rational number $r=2^i$ such that $i\in \IZ$.
Our navigating net records a sequence of $r$-nets 
$Y_{r_{\min}}, Y_{2r_{\min}},\dots,Y_{r_{\max}}$
such that $r_{\min}$ and $r_{\max}$ are scales satisfying the inequalities
$\dmin(Y)/2 < r_{\min} \leq \dmin(Y)$ and $\diam(Y)/2 \leq r_{\max} < 2\diam(Y)$.
At the largest scale $r_{\max}$, the $r_{\max}$-net $Y_{r_{\max}}$ consists
of a single point $y_{\mathrm{root}}$, and thus $Y=\ball(y_{\mathrm{root}},r_{\max})$.
At the lowest scale, we set $Y_{r_{\min}}=Y$.
So the navigating net represents $Y$ at all scales $r$ such
that $r_{\min} \leq r \leq r_{\max}$ by an $r$-net $Y_r$.
All scales $r > r_{\max}$ are represented by
$Y_r=Y_{r_{\max}}=\{y_{\mathrm{root}}\}$, but we do not construct
these copies of $Y_{r_{\max}}$ explicitly.

We may assume that $\log(\Phi(Y))\geq 1$ as otherwise,
$|Y|=O(1)$ by Lemma~\ref{lem:packing}, and the $\rho$-distortion
problem can be solved in $O(1)$ time by brute force.
So we only construct $Y_r$ at  $O(\log(\Phi(Y))$ different scales $r$
such that $\dmin(Y)/2 < r_{\min} \leq r \leq r_{\max} <2 \diam(Y)$.

We construct a graph over these $r$-nets. First, at each scale $r$ such
that $r_{\min} \leq r \leq r_{\max}$,
we connect any two nodes $y$, $y' \in Y_r$ such that 
$d_Y(y,y') \leq 6r$ by an edge, that we call a {\em horizontal} edge.
At each scale $r$ such that $2r_{\min} \leq r \leq r_{\max}$,
we also connect with a {\em vertical} edge each $y \in Y_{r/2}$ to a 
node $p(y) \in Y_{r}$,
called the {\it parent} of $y$, such that $d_Y(y,p(y)) \leq r$.
At least one such node $p(y)$ exists since $Y_r$ is an $r$-net.
We call $y$ a {\it child} of $p(y)$. More generally, 
we say that $y$ is a {\it descendant} of $y'$ if $y$
is a child of $y'$, or $y$ is a child of a descendant of $y'$.
Conversely, we say that $y'$ is an {\em ancestor} of $y$ if
$y$ is a descendant of $y'$. 
When $r'>r_{\max}$, the node $y_{\mathrm{root}} \in Y_{r'}$ is an 
ancestor of any node at a lower level. 

These $r$-nets, together with the horizontal and vertical edges,
form our navigating net. 
The lemma below shows that it has bounded degree. 
\begin{lemma}\label{lem:sizennet}
	Each node of the navigating net has degree $O(1)$.
\end{lemma}
\begin{proof}
	The children of $y \in Y_r$ are at distance at most
	$r$ from $y$, and thus at distance at most $2r$ from each
	other. As they are in $Y_{r/2}$, they 
	are at distance at least $r/2$ from each other. 
	So the spread of the set of children of $y$ is at most 4, thus
	$y$ has at most $16^\delta$ children by Lemma~\ref{lem:packing}.
	It follows that any node is incident to at most 
	$1+16^\delta=O(1)$ vertical edges.
	The nodes adjacent to $y$ via horizontal edges are at 
	distance at most $6r$ from $y$ and at least $r$ from
	each other. Hence their spread is at most $12$, so
	there are at most $48^{\delta}$ of them. 
\end{proof}

The lemma below shows that our navigating net can be computed 
efficiently. 

\begin{lemma}\label{lem:timennet}
	The navigating net of $Y$ can be computed in 
	$O(n \log \Phi(Y))$ time.
\end{lemma}
\begin{proof}
	We construct the navigating net incrementally,
	adding the nodes one-by-one in a top-down manner.	
	We first pick an arbitrary $y_{\mathrm{root}} \in Y$,
	and we let $r_{\max}$ be the smallest scale that
	is at least $\max_{y \in Y}d_Y(y_{\mathrm{root}},y)$.
	We initialize our navigating net with the single node
	$y_{\mathrm{root}}$ at scale $r_{\max}$.
	Initially, we set $r_{\min}=r_{\max}$, and we will decrease
	$r_{\min}$ when we insert new points, if needed. We will
	maintain the invariants that any two inserted points
	are at distance at least $r_{\min}$ from each other,
	and that $Y_{r_{\min}}$ is the set of all inserted points.

	We repeatedly insert an arbitrary point $y \in Y$,
	until all points have been inserted. We now show how to
	update the data structure while inserting $y$.

	We start from the root and move down to scale $r_{\min}$.
	During this process, at the current scale $r$, 
	we maintain the set $N(y,r)$ of points in $Y_r$	that are at 
	distance at most 	$6r$ from $y$. We have $\dmin(N(y,r)) \geq r$
	and $\diam(N(y,r))\leq 12r$, hence $|N(y,r)|\leq 48^\delta=O(1)$ 
	by Lemma~\ref{lem:packing}.

	We now show how to maintain $N(y,r)$ while traversing the
	data structure from top to bottom. So we assume that we know
	$N(y,2r)$, and we show how to find $N(y,r)$. Suppose that
	$y' \in N(y,r)$. Then we have $d_Y(y',y) \leq 6r$ and
	$d_Y(y',p(y')) \leq 2r$. Therefore, $d_Y(y,p(y'))\leq 8r$,
	which implies that $p(y') \in N(y,2r)$. So in order to find
	the points in $N(y,r)$, we can simply go through all the
	children of the nodes in $N(y,2r)$, and check for each of
	them if it is at distance at most $6r$ from $y$. It takes
	$O(1)$ time as the navigating net has constant degree.
	So we can compute all the sets $N(y,r)$ such that
	$r_{\min} \leq r \leq r_{\max}$ in $O(\log \Phi(Y))$ time
	as there are $O(\log \Phi(Y))$ different scales.
	
	If $N(y,r_{\min})$ is empty, then the closest inserted point $y_c$
	to $y$ is at distance more than $r_{\min}$, so we don't
	need to update $r_{\min}$. Otherwise, since $y_c \in N(y,r_{\,min})$,
	we can find $y_c$ in $O(1)$ time by brute force. 
	If $d_Y(y,y_c) < r_{\min}$,
	we set $r_{\mathrm{temp}}$ to be the largest scale that is
	at most $d_Y(y,y_c)$, and we create new levels in the navigating
	net at all scales $r$ such that $r_{\mathrm{temp}} \leq r < r_{\min}$.
	As $Y_{r_{\min}}$ consists of all the inserted points, the
	same is true for $Y_r$ at the new levels, and we can
	compute the horizontal edges in $O(1)$ time per node as
	they are a subset of the horizontal edges at scale $2r$.
	We also connect each node by a vertical edge to its copy at
	scale $2r$. After updating all levels, we set $r_{\min}=r_{\mathrm{temp}}$.
	
	We now show how to insert $y$ in the navigating net. When
	inserting $y$ at level $r$, we check whether all the points
	in $N(y,r)$ are at distance at distance at least  $r$
	from $y$. If it is the case, we insert $y$ into $Y_r$, and
	we connect $y$ to all the points in $N(y,r)$ by a horizontal
	edge.  Then let $y''$ be the point of
	$N(y,2r)$ that is closest to $y$. As $Y_{2r}$ is a $2r$-net,
	we have $d_Y(y,y'') \leq 2r$, so we set $y''=p(y)$.

	This algorithm constructs the navigating net incrementally
	in $O(\log \Phi(Y))$ time per point in $Y$, so the overall 
	running time 	is $O(n \log \Phi(Y))$. Our construction ensures
	that $\dmin(Y)/2 < r_{\min} \leq \dmin(Y)$ and $\diam(Y)/2\leq r_{\max} < 2\diam(Y)$.
	
\end{proof}

We associate the ball $\ball(y,2r)$ with each node $y$ in an $r$-net $Y_r$. 
These balls have the following properties.
\begin{lemma}\label{lem:balldiameter}
	At any scale $r$, and for any subset $S \subseteq Y$ such that
	$\diam(S) \leq r$, there exists a node $y \in Y_r$ such
	that $S \subseteq \ball(y,2r)$.
\end{lemma}
\begin{proof}
		Let $z$ be a point in $S$. Since $Y_r$ is an $r$-net, 
	there exists $y \in Y_r$ such that $d_Y(y,z) \leq r$.
	For any $z' \in S$, we have $d_Y(z,z') \leq \diam(S) \leq r$ 
	so by the triangle
	inequality, $d_Y(y,z') \leq 2r$. In other words, $S \subseteq \ball(y,2r)$.
\end{proof}
The lemma below shows that any ball other than the root is contained
in the ball associated with any of its ancestors.
\begin{lemma}\label{lem:ballcontainment}
	Let $r$ and $r'$ be two scales such that 
	$r < r'$. For any $y \in  Y_{r}$,
	and for any ancestor $y' \in Y_{r'}$ of $y$,
	we have $b(y,2r) \subseteq \ball(y',2r')$.
\end{lemma}
\begin{proof}
	When $r'=2r$,  we have $y'=p(y)$, so $d_Y(y,y') \leq r'$ 
	by definition of $p(y)$. It follows that 
	$\ball(y,2r) \subseteq \ball(y',2r+r')=\ball(y',2r')$.
	When $r' \geq 4r$, the result is obtained by induction.
\end{proof}
In summary, the balls $\ball(y,2r)$, connected
by the vertical edges, form a tree such that
each ball is contained in each of its ancestors. The horizontal edges will
help us traverse this tree within a given level.

\subsection{Splitting the pattern}\label{sec:split}

Our algorithm proceeds recursively, by partitioning $X$ into
two well-separated subsets $P$ and $Q$ at each stage. 
More precisely, we will split $X$ as follows. (Remember that $k=|X|$.)
\begin{lemma}\label{lem:splitting}
	If $k\geq 2$, we can partition $X$ into two non-empty subsets
	$P$ and $Q$ such that $\diam(X) \leq (k-1) \cdot d_X(P,Q)$.
\end{lemma}
\begin{proof}
	We obtain $P$ and $Q$ by running Kruskal's algorithm~\cite{CLRS} 
	for computing a minimum spanning tree of $X$, 	
	and stopping at the second-last step.
	So starting from the forest $(X,\emptyset)$, we repeatedly
	insert the shortest edge that connects any two trees of the current forest,
	until we are left with exactly two trees $P$ and $Q$.	
	At the last step, $P$ and $Q$ are then connected with an edge of length 
	$\ell=d_X(P,Q)$,
	which is not shorter than any edge in the spanning trees we 
	constructed for $P$ and $Q$.
	Thus, $\diam(P)\leq (|P|-1)\ell$ and $\diam(Q)	\leq (|Q|-1)\ell$.
	It follows that 
	$\diam(X) \leq   (|P|-1)\ell + \ell+ (|Q|-1)\ell=(k-1)\ell$. 
\end{proof}

\subsection{Recording approximate matchings}\label{sec:recordingam}

Our algorithm for the $(\rho,\eps)$-distortion problem
records a collection of approximate matchings for some
layers of the navigating net. More precisely, for some 
non-empty subset $W$ of $X$, for some $0<\beta \leq 1$ and 
for some scales $r \geq \rho\diam(W)$,
we will construct a data structure $L(W,\beta,r)$ that 
records at least one $(1+\beta)\rho$-matching from $W$ to $Y$ 
if a $\rho$-matching from $W$ to $Y$
exists. In particular, we will compute $L(X,\eps,r_X)$ at an appropriate
scale $r_X$, which records a solution to the $(\rho,\eps)$-distortion problem if there is one.
We first give three invariants of $L(W,\beta,r)$, and in the
next two sections we will show how to compute
$L(X,\eps,r_X)$  recursively.

The data structure $L(W,\beta,r)$ records
a set $M(y,W,\beta,r)$ of matchings at each node $y \in Y_r$.
These sets satisfy the following properties.
\begin{property}\label{prop:DXY}
Let $r$ be a scale such that $r \geq \rho \diam(W)$.
\begin{enumerate}
	\item[(a)] For any $y \in Y_r$, each matching $\sigma_\beta \in M(y,W,\beta,r)$ 
		is a $(1+\beta)\rho$-matching from $W$ to $Y$ such that
		$\sigma_\beta(W) \subseteq \ball(y,3r)$.	
	\item[(b)] For any  $y \in Y_r$ and  
		any two distinct $\sigma_\beta$, $\sigma_\beta' \in M(y,W,\beta,r)$, we have
		$d_M(\sigma_\beta,\sigma_\beta') \geq \beta r/(2\rho^2)$.
	\item[(c)] For any $\rho$-matching $\sigma:W \to Y$,
		there exist $y \in Y_r$ and  $\sigma_\beta \in M(y,W,\beta,r)$
		such 	that $d_M(\sigma,\sigma_\beta) \leq \beta r/\rho^2$.
\end{enumerate}
\end{property}

These properties imply the following bound on the sizes of these sets.
\begin{lemma}\label{lem:sizemr}
	For any $y\in Y_r$, we have  $|M(y,W,\beta,r)|=O((\rho^2/\beta)^{k\delta})$.
\end{lemma}
\begin{proof}
	Property~\ref{prop:DXY}a implies that any two matchings
	in $M(y,W,\beta,r)$ are at distance at most $6r$ from each other,
	that is, $\diam(M(y,W,\beta,r))\leq 6r$.
	Property~\ref{prop:DXY}b means that $\dmin(M(y,W,\beta,r)) \geq \beta r/(2\rho^2)$.
	Therefore, $\Phi(M(y,W,\beta,r)) \leq 12\rho^2/\beta$. Then by 
	Lemma~\ref{lem:packing} and Corollary~\ref{cor:dimmatching},
	we have $|M(y,W,\beta,r)| \leq (48\rho^2/\beta)^{k\delta}$.
\end{proof}

\subsection{Recursive construction} \label{sec:recursivec}
Our algorithm constructs $L(W,\beta,r)$ recursively,
for some subsets $W$ of $X$ and some values $\beta$ and $r$.
We start with the base case, then we show how to compute $L(W,\beta,r')$
from $L(W,\beta,r)$ when $\rho\diam(W) \leq r < r'$.
Finally, we show how to recursively compute $L(X,\eps,r_X)$ at an
appropriate scale $r_X$ by splitting $X$ into $P$ and $Q$ according 
to Lemma~\ref{lem:splitting}. (See \figurename~\ref{fig:split}.)

\begin{figure}[h]
	\centering
	\includegraphics[scale=.75]{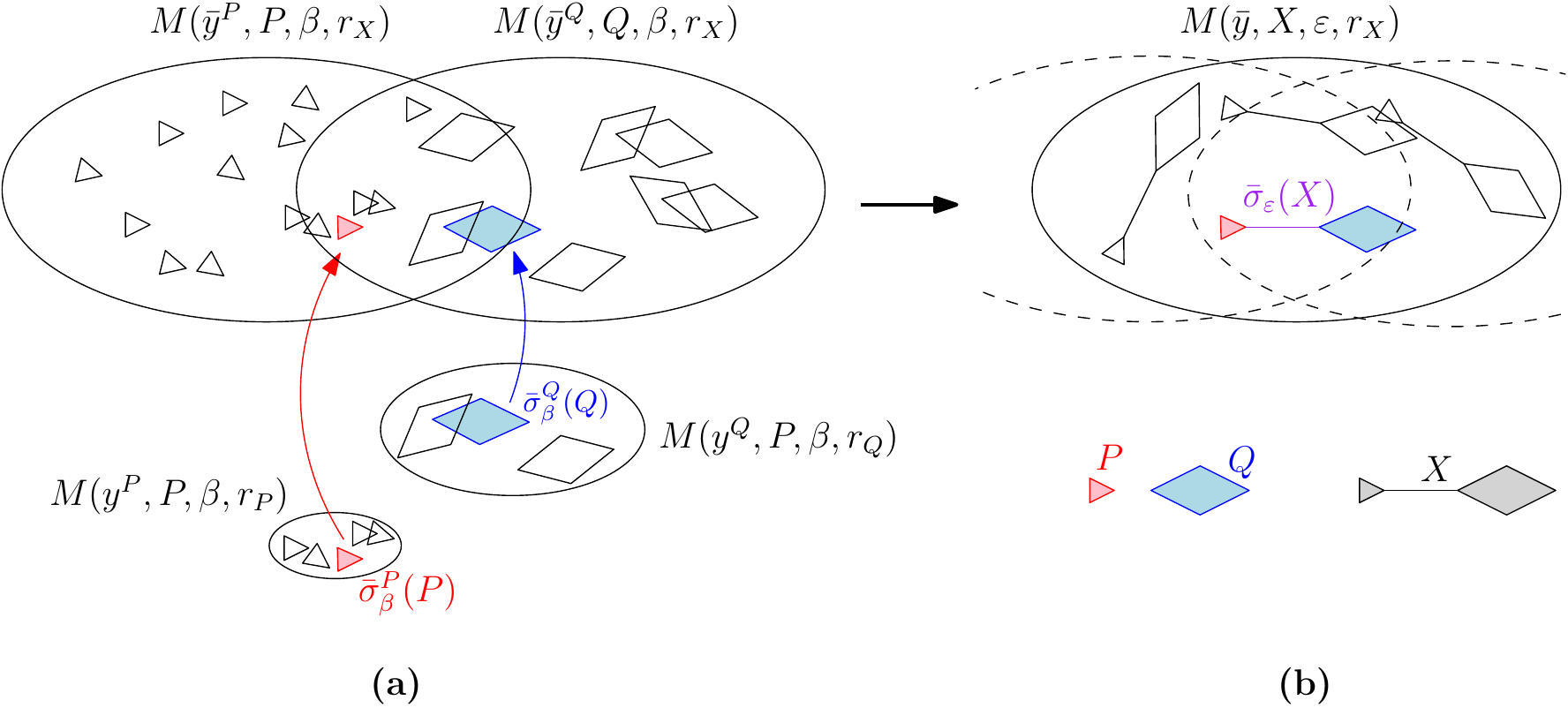}
	\caption{Recursive construction of $L(X,\eps,r_X)$.
		The pattern $X$ is split into $P$ and $Q$. 
		(a) Bottom-up phase: We compute $L(P,\beta,r_X)$ and $L(Q,\beta,r_X)$
		from $L(P,\beta,r_P)$ and $L(P,\beta,r_Q)$.
		(b) An approximate
		matching $\bar \sigma_\eps:X \to Y$ in $L(X,\eps,r_X)$ is obtained by 
		combining  two matchings  $\bar \sigma_\beta^P:P \to Y$ and
		$\bar \sigma_\beta^Q:Q \to Y$. 
		\label{fig:split}}
\end{figure}

\paragraph*{Base case where $|W|=1$ and $r=r_{\min}$.}
We have $W=\{x_i\}$, so a $(1+\beta)\rho$-matching simply maps $x_i$
to any element of $Y$. Therefore, at scale $r=r_{\min}$, we have $Y_r=Y$, 
so for each $y \in Y$, we  record the matching that sends $x_i$ to $y$ 
in $M(y,W,\beta,r_{\min})$.

\paragraph*{Bottom-up construction.} 
Suppose that $L(W,\beta,r)$
has been computed for some scale $r \geq \rho \diam(W)$. Given a larger
scale $r' > r$, we now show how to construct $L(W,\beta,r')$.
(See Algorithm~\ref{alg:bottomup}.)

\begin{algorithm}[H]
	\caption{Computing $L(W,\beta,r')$ from  $L(W,\beta,r)$} \label{alg:bottomup}
	\begin{algorithmic}[1]
		\Procedure{BottomUpConstruction}{}
			\For{$y' \in Y_{r'}$}
				\State $M(y',W,\beta,r') \gets \emptyset$
				\For{each descendant $y$ of $y'$ at scale $r$}
						\For{each matching $\sigma_\beta \in M(y,W,\beta,r)$}
								\If{$d_M(\sigma_\beta,\sigma'_\beta) \geq \beta r'/(2\rho^2)$
									for all $\sigma'_\beta \in M(y',W,\beta,r')$}
									\State insert $\sigma_\beta$ into $M(y',W,\beta,r')$
								\EndIf
					\EndFor
				\EndFor
			\EndFor
		\EndProcedure
	\end{algorithmic}
\end{algorithm}

For each $y' \in Y_{r'}$, we proceed as follows. 
Initially, we set $M(y',W,\beta,r')=\emptyset$.
Then for each descendant $y \in Y_r$ of $y'$, and for each 
matching $\sigma_\beta \in M(y,W,\beta,r)$,
we check by brute force whether 
$d_M(\sigma_\beta,\sigma_\beta') \geq \beta r'/(2\rho^2)$ for
each matching $\sigma_\beta'$
that was previously inserted into  $M(y',W,\beta,r')$. If it is the case,
we insert $\sigma_\beta$ into $M(y',W,\beta,r')$, and otherwise we discard
$\sigma_\beta$.
It ensures that  Property~\ref{prop:DXY}b holds for $M(y',W,\beta,r')$.

We now prove that Property~\ref{prop:DXY}a holds.
Let $\sigma_\beta \in M(y',W,\beta,r')$. By construction, 
$\sigma_\beta \in M(y,W,\beta,r)$ for some descendant $y$ of $y'$.
As Property~\ref{prop:DXY}a holds at scale $r$, it
follows that $\sigma_\beta$ is a $(1+\beta)\rho$-matching, and that
$\sigma_\beta(W) \subseteq \ball(y,3r)$. Since $y$ is a descendant
of $y'$, by Lemma~\ref{lem:ballcontainment}, we have 
$\ball(y,2r) \subseteq \ball(y',2r')$, and thus 
$\ball(y,3r) \subseteq \ball(y',2r'+r) \subseteq \ball(y',3r')$.
It follows that $\sigma_\beta(W) \subseteq \ball(y',3r')$,
and thus Property~\ref{prop:DXY}a holds for $M(y',W,\beta,r')$.

Finally, we prove that Property~\ref{prop:DXY}c holds
as well. Let $\sigma:W \to Y$ be a $\rho$-matching. 
As Property~\ref{prop:DXY}c holds for $L(W,\beta,r)$,
there must be a node $y \in Y_{r}$ and $\sigma_\beta \in M(y,W,\beta,r)$
such that $d_M(\sigma,\sigma_\beta)\leq \beta r/\rho^2$.
Let $y'\in Y_{r'}$ be the ancestor of $y$ at scale $r'$.
If $\sigma_\beta$ was inserted into $M(y',W,\beta,r')$, then
we are done because 
$d_M(\sigma,\sigma_{\beta}) \leq \beta r/\rho^2 \leq \beta r'/\rho^2$. 
Otherwise, we must have inserted
a matching $\sigma_\beta'$ such that 
$d_M(\sigma_\beta,\sigma_\beta') < \beta r'/(2\rho^2)$.
It follows that $d_M(\sigma,\sigma_\beta') \leq d_M(\sigma,\sigma_\beta)
+d_M(\sigma_\beta,\sigma_\beta') \leq 
\beta (r+r'/2)/\rho^2 \leq \beta r'/\rho^2$,
which completes the proof that Property~\ref{prop:DXY} holds
for $L(W,\beta,r')$.

We now analyze this algorithm. First we need to find all the descendants
of each node $y' \in Y_{r'}$. We can do this by traversing the
navigating net, which takes time $O(n \log\Phi(Y))$ as there are
$O(\log\Phi(Y))$ levels in the navigating net.
By Lemma~\ref{lem:sizemr}, we have
$|M(y',W,\beta,r')|=O((\rho^2/\beta)^{k\delta})$. Therefore, each time we attempt
to insert a matching $\sigma_\beta$ into $M(y',W,\beta,r')$, we compare it with
$O((\rho^2/\beta)^{k\delta})$ previously inserted matchings,
so it takes $O(k(\rho^2/\beta)^{k\delta})$ time as the distance between
two matchings can be computed in $O(k)$ time. Since $|Y_r| \leq n$, 
and each set $M(y,W,\beta,r)$ has cardinality 
$O((\rho^2/\beta)^{k\delta})$, we spend $O(nk(\rho^2/\beta)^{2k\delta})$ time
for computing $L(W,\beta,r')$. So we just
proved the following.
\begin{lemma}\label{lem:bottom-up}
	Let $r$ and $r'$ be two scales such that $\rho\diam(W) \leq r <  r'$.
	Given $L(W,\beta,r)$, we can compute 	$L(W,\beta,r')$ in 
	$O(nk(\rho^2/\beta)^{2k\delta}+n\log \Phi(Y))$ time.
\end{lemma}

\paragraph*{Computing $L(X,\eps,r_X)$ by  splitting $X$.} 

Suppose that $k \geq 2$.
Let $r_X$ be the smallest scale that is at least as large as $\rho\diam(X)$,
hence $r_X=2^{\lceil \log_2(\rho\diam(X)) \rceil}$.
In particular, we have $r_X/2 < \rho\diam(X) \leq r_X$.
If $r_X < r_{\min}$, then we have $\rho \diam(X) < r_{\min} \leq \dmin(Y)$, and 
there cannot be any $\rho$-matching, so our algorithm does not return
any matching. Therefore, from now on, we may assume that $r_X \geq r_{\min}$.

Let $P$ and $Q$ be the sets obtained by splitting $X$ as 
described in Lemma~\ref{lem:splitting}, so
$\diam(X) \leq (k-1)\cdot d_X(P,Q)$.
Let $\beta=\eps/(8k-8)$, and suppose that $L(P,\beta,r_X)$
and $L(Q,\beta,r_X)$ have been computed earlier.
We now show how to compute $L(X,\eps,r_X)$.

For any two matchings $\sigma^P:P \to Y$ and $\sigma^Q: Q \to Y$,
we denote by $\sigma^P\cdot \sigma^Q$ the matching from $X$
to $Y$ whose restrictions to $P$ and $Q$ are $\sigma^P$ and $\sigma^Q$,
respectively. In other words, if $\sigma=\sigma^P\cdot \sigma^Q$, 
then we have $\sigma(x)=\sigma^P(x)$ for
all $x \in P$ and $\sigma(x)=\sigma^Q(x)$ for all $x \in Q$.

We compute $L(X,\eps,r_X)$ as follows. (See Algorithm~\ref{alg:lowlevel}.) 
For each node $y \in Y_{r_X}$,
we consider all the pairs of matchings consisting of a matching
$\sigma^P_\beta \in M(y^P,P,\beta,r_X)$ and 
a matching $\sigma^Q_\beta \in M(y^Q,Q,\beta,r_X)$,
where $y^P$ and $y^Q$ are in $Y_{r_X}$ and are at distance at most $6r_X$ from $y$.
Then we consider $\sigma_\eps=\sigma^P_\beta\cdot\sigma^Q_\beta$
as a candidate for being inserted into $M(y,X,\eps,r_X)$. 
We first check whether $\sigma_\eps$
is a $(1+\eps)\rho$-matching 
and $\sigma_\eps(X) \subseteq \ball(y,3r_X)$. 
If it is the case, and if $\sigma_\eps$ is
at distance at least $\eps r_X/(2\rho^2)$ from any matching previously
inserted into $M(y,X,\eps,r_X) $, we insert $\sigma_\eps$ into $M(y,X,\eps,r_X) $.

\begin{algorithm}[H]
	\caption{Computing $L(X,\eps,r_X)$ from 
		$L(P,\beta,r_X)$ and $L(Q,\beta,r_X)$} \label{alg:lowlevel}
	\begin{algorithmic}[1]
		\Procedure{ComputeBySplitting}{}
			\For{$y \in Y_{r_X}$}
				\State $M(y,X,\eps,r_X) \gets \emptyset$
				\For{each pair of horizontal edges $(y,y^P)$, $(y,y^Q)$}
						\For{each pair $\sigma^P_\beta \in M(y^P,P,\beta,r_X)$,
							$\sigma^Q_\beta \in M(y^Q,Q,\beta,r_X)$}
							\State $\sigma_\eps \gets 
								\sigma_\beta^P \cdot \sigma_\beta^Q$
							\If{$\sigma_\eps$ is a $(1+\eps)\rho$-matching
									and $\sigma_\eps(X) \subseteq \ball(y,3r_X)$}
								\If{$d_M(\sigma_\eps,\sigma'_\eps) \geq \eps r_X/(2\rho^2)$
									for all $\sigma'_\eps \in M(y,X,\eps,r_X)$}
									\label{line:lowlevel:compare}
									\State insert $\sigma_\eps$ into $M(y,X,\eps,r_X)$
								\EndIf
							\EndIf
					\EndFor
				\EndFor
			\EndFor
		\EndProcedure
	\end{algorithmic}
\end{algorithm}

We first prove that this algorithm is correct. So we must
prove that Property~\ref{prop:DXY} holds for each set $M(y,X,\eps,r_X)$.
Property~\ref{prop:DXY}a  follows from the fact that
we only insert $\sigma_\eps$ if it 
is a $(1+\eps)\rho$-matching and if 
$\sigma_\eps(X) \subseteq \ball(y,3r_X)$.
Property~\ref{prop:DXY}b follows from fact that we only
insert $\sigma_\eps$ if it is at distance at least
$\eps r_X/(2\rho^2)$ from all the previously inserted matchings.

We now prove that Property~\ref{prop:DXY}c holds.
So given a $\rho$-matching $\bar \sigma:X \to Y$, we want to prove
that there exist $\bar y \in Y_{r_X}$ and  
$\bar\sigma_\eps \in M(\bar y,X,\eps,r_X)$
such 	that $d_M(\bar\sigma,\bar\sigma_\eps) \leq \eps r_X/\rho^2$.
Let 
$\bar\sigma^P$ be the restriction of $\bar \sigma$ to $P$. In other words,
$\bar\sigma^P:P \to Y$ is defined by $\bar\sigma^P(x)=\bar\sigma(x)$ 
for all $x \in P$. 
Similarly, let $\bar\sigma^Q$ be the restriction of $\bar\sigma$ to $Q$.
Then $\bar\sigma^P$ and $\bar\sigma^Q$ are $\rho$-matchings.
As Property~\ref{prop:DXY}c holds for $L(P,\beta,r_X)$,
there exist $\bar y^P \in Y_{r_X}$ 
and $\bar\sigma^P_{\beta} \in M(\bar y^P,P,\beta,r_X)$
such that $d_Y(\bar\sigma^P,\bar\sigma^P_{\beta})\leq \beta r_X/\rho^2$.
Similarly, there exist $\bar y^Q \in Y_{r_X}$ and  
$\bar\sigma^Q_{\beta} \in M(\bar y^Q,Q,\beta,r_X)$
such that $d_Y(\bar\sigma^Q,\bar\sigma^Q_{\beta})\leq \beta r_X/\rho^2$.

The matching $\bar\sigma_\eps=\bar\sigma^P_\beta \cdot \bar\sigma^Q_\beta$ 
satisfies 
$d_M(\bar\sigma,\bar\sigma_\eps)\leq \beta r_X/\rho^2 < \eps r_X/\rho^2$.
The lemma below shows that it is a $(1+\eps)\rho$-matching.
\begin{lemma}\label{lem:apxcombine}
	The matching $\bar\sigma_\eps$ is a $(1+\eps)\rho$-matching 
	from $X$ to $Y$.
\end{lemma}
\begin{proof}
	Let $x, x' \in X$. If $(x,x') \in P^2$ or $(x,x') \in Q^2$,
	then 
	\[\frac{1}{(1+\beta)\rho} d_X(x,x') \leq d_Y(\bar\sigma_\eps(x),
		\bar\sigma_\eps(x')) \leq (1+\beta)\rho d_X(x,x')\] 
	since $\bar\sigma^P_\beta$ and $\bar\sigma^Q_\beta$ are 
	$(1+\beta)\rho$-matchings by Property~\ref{prop:DXY}a. 
	As $\beta<\eps$, it follows that
	\[\frac{1}{(1+\eps)\rho} d_X(x,x') \leq d_Y(\bar\sigma_\eps(x),
		\bar\sigma_\eps(x')) \leq (1+\eps)\rho d_X(x,x')\] 
	which is the desired inequality.	

	So we may now	assume, without loss of generality, 
	that $x \in P$ and $x' \in Q$,
	and thus $d_X(x,x') \geq d_X(P,Q)$.
	As $X$ was partitioned into $P$ and $Q$ using Lemma~\ref{lem:splitting},
	we have  $\diam(X) \leq (k-1)\cdot d_X(P,Q)$. By the definition of
	$r_X$, we also have 	$\rho\diam(X) > r_X/2$, and thus
	$(k-1)d_X(x,x') > r_X/(2\rho)$. By the triangle inequality, we have
	$
		d_Y(\bar\sigma_\eps(x),\bar\sigma_\eps(x')) 
		\leq d_Y(\bar\sigma_\eps(x),\bar\sigma(x))+
			d_Y(\bar\sigma(x),\bar\sigma(x')) +
			d_Y(\bar\sigma(x'),\bar\sigma_\eps(x')),
	$
	and thus
	\begin{align*}
		d_Y(\bar\sigma_\eps(x),\bar\sigma_\eps(x')) 
		& \leq d_Y(\bar\sigma(x),\bar\sigma(x'))+\frac{2\beta r_X}{\rho^2} &
			\text{because $d_M(\bar\sigma,\bar\sigma_\eps) \leq  \frac{\beta r_X}{\rho^2}$} \\
		& \leq 	\rho d_X(x,x')+\frac{2\beta r_X}{\rho^2}	& 
			\text{since $\bar\sigma$ is a $\rho$-matching} \\	
		& \leq \rho d_X(x,x')+\frac{4\beta(k-1)}{\rho}d_X(x,x') & \\
 		& \leq 	\rho d_X(x,x')+\frac{\eps}{2\rho} d_X(x,x')
 			& \text{because $\beta = \frac{\eps}{8k-8}$} \\	
		& \leq \left(1+\frac{\eps}{2}\right)\rho  d_X(x,x') .
	\end{align*}
	
	We just proved the right-hand side of Inequality~\eqref{eq:distortion}.
	For the left-hand side, we use the same arguments as above,
	and obtain
	\begin{align*}
		d_Y(\bar\sigma_\eps(x),\bar\sigma_\eps(x')) & 
			\geq d_Y(\bar\sigma(x),\bar\sigma(x'))-\frac{2\beta r_X}{\rho^2}
			 \geq 	\frac{1}{\rho}d_X(x,x')-\frac{2\beta r_X}{\rho^2} 	\\
			 & \geq \frac{1}{\rho} d_X(x,x')-\frac{4\beta(k-1)}{\rho} d_X(x,x') \\
			 &  =\frac{1}{\rho} d_X(x,x')-\frac{\eps}{2\rho}d_X(x,x')
			 	= \left(1 -\frac \eps{2}\right) \frac{1}{\rho}d_X(x,x').
	\end{align*}
	So in order to complete the proof, we only need to argue that
	$1-\eps/2 \geq 1/(1+\eps).$
	This is equivalent to 
	$1+\eps/2 -\eps^2/2 \geq 1$,
	which follows from our assumptions that $0<\eps \leq 1$.
\end{proof}

As $\bar \sigma$ is a $\rho$-matching, we have 
$\diam(\bar \sigma(X)) \leq \rho \diam(X)$ and thus $\diam(\bar \sigma(X)) \leq r_X$.
By Lemma~\ref{lem:balldiameter}, it implies that  
there exists  $\bar y\in Y_{r_X}$ such 
that $\bar\sigma(X) \subseteq \ball(\bar y,2r_X)$, and thus
$\bar\sigma^P(P) \subseteq \ball(\bar y,2r_X)$. As  
$d_M(\bar\sigma^P,\bar\sigma^P_{\beta})\leq \beta r_X/\rho^2 \leq r_X/8$, we have
$\bar\sigma^P_{\beta}(P)\subseteq \ball(\bar y,17r_X/8)$. 
By Property~\ref{prop:DXY}a, we also have
$\bar\sigma^P_\beta(P) \subseteq \ball(\bar y^P,3r_X)$,
so the balls 
$\ball(\bar y,17r_X/8)$ and $\ball(\bar y^P,3r_X)$ intersect, which implies 
that $d_Y(\bar y,\bar y^P) < 6r_X$. 
The same proof shows that $d_Y(\bar y,\bar y^Q)  <6r_X$.

Therefore, our algorithm considers 
$\sigma_\eps=\sigma_\beta^P \cdot \sigma_\beta^Q$ as a candidate 
solution for each $\sigma_\beta^P \in M(\bar y^P,P,\beta,r_X)$ and each 
$\sigma_\beta^Q \in M(\bar y^Q,Q,\beta,r_X)$. In particular, 
we must have considered the matching
$\bar\sigma_\eps=\bar\sigma^P_\beta \cdot \bar\sigma^Q_\beta$. 
As $d_M(\bar \sigma,\bar \sigma_\eps) < \eps  r_X / \rho^2$ and
$\bar\sigma(X)\subseteq b(\bar y,2r_X)$, we have 
$\bar\sigma_\eps(X) \subseteq b(\bar y,3r_X)$, and thus we
must have attempted to insert $\bar\sigma_\eps$ into 
$M(\bar y,X,\eps,r_X)$.
If $\bar \sigma_\eps$ was inserted into $M(\bar y,X,\eps,r_X)$, 
then we are done.
Otherwise, it means that there exists 
$\sigma_\eps' \in M(\bar y,X,\eps,r_X)$ 
such that $d_M(\bar\sigma_\eps,\sigma_\eps')< \eps r_X/(2\rho^2)$.
Since 
$d_M(\bar\sigma,\bar\sigma_\eps)\leq  \beta r_X/\rho^2<\eps r_X/(2\rho^2)$, 
it follows
that $d_M(\bar \sigma,\sigma_\eps') < \eps r_X/\rho^2$.
In any case, it shows that Property~\ref{prop:DXY}c holds.
So we obtain the following result.
\begin{lemma}\label{lem:lowlevel}
	Suppose that $k \geq 2$ and $\beta=\eps/(8k-8)$. Then
	we can compute $L(X,\eps,r_X)$ from  
	$L(P,\beta,r_X)$ and $L(Q,\beta,r_X)$ in  
	$O(nk(\rho^2/\beta)^{3k\delta})$ time.
\end{lemma}
\begin{proof}
	The discussion above shows that our algorithm is
	correct. We still need to analyze its running time.
	Let $y \in Y_{r_X}$. As $d_Y(y,y^P)\leq 6r_X$ and $d_Y(y,y^Q) \leq 6r_X$,
	the nodes $y^P$ and $y^Q$ are connected to $y$ by horizontal edges.
	As the nodes of the navigating net have constant degree,
	it implies that there are $O(1)$ pairs $(y^P,y^Q)$ to consider. 
	By Lemma~\ref{lem:sizemr}, there are
	$O((\rho^2/\beta)^{k\delta})$ matchings in $M(y^P,P,\beta,r_X)$ and
	$M(y^Q,Q,\beta,r_X)$. 
	Therefore, we consider $O((\rho^2/\beta)^{2k\delta})$ matchings
	$\sigma_\eps$ when constructing $M(y,X,\eps,r_X)$. 
	Each of these matchings is then compared with 
	previously inserted matchings.
	By Lemma~\ref{lem:sizemr}, we have 
	$|M(y,X,\eps,r_X)|=O((\rho^2/\eps)^{k\delta})=O((\rho^2/\beta)^{k\delta})$.
	We can compute the distance between two matchings in $O(k)$ time, 
	so it takes  
	$O(k(\rho^2/\beta)^{3k\delta})$ time to
	compute $M(y,X,\eps,r_X)$.
	As there are at most $n$ nodes $y \in Y_{r_X}$, the overall time bound
	is $O(nk(\rho^2/\beta)^{3k\delta})$.
\end{proof}

\subsection{Putting everything together} \label{sec:main}
We can now describe our algorithm 
for the $(\rho,\eps)$-distortion problem. 
We first compute the navigating net in 
$O(n \log \Phi(Y))$ time.
If $r_X < r_{\min}$, then $\rho \diam X < \dmin(Y)$, and thus
there is no $\rho$-matching.
Otherwise, we recursively compute $L(X,\eps,r_X)$,
as described below. If $L(X,\eps,r_X)$ contains at least one matching, 
then we return one of them, which by Property~\ref{prop:DXY}a is a 
$(1+\eps)\rho$-matching.
By Property~\ref{prop:DXY}c, if
there is a solution to the $\rho$-distortion problem, then  at least one  
$(1+\eps)\rho$-matching must be recorded in $L(X,\eps,r_X)$, which
shows that this algorithm indeed solves the 
$(\rho,\eps)$-distortion problem.

We now explain how to compute $L(X,\eps,r_X)$ recursively.
The base case is $k=1$ and $r_X=r_{\min}$. As explained above,
it can be done in $O(n)$ time by recording a trivial matching
at each leaf node.

When $k \geq 2$, we split $X$ into $P$ and $Q$ as described in
Lemma~\ref{lem:splitting}. Then we compute recursively 
$L(P,\beta,r_P)$ and $L(Q,\beta,r_Q)$ where 
$\beta=\eps/(8k-8)$ and $r_P$ and $r_Q$ are the smallest
scales at least as large as $\rho\diam(P)$ and $\rho\diam(Q)$, 
respectively. We compute $L(P,\beta,r_X)$ and
$L(Q,\beta,r_X)$ using Lemma~\ref{lem:bottom-up},
which takes time  $O(nk(\rho^2/\beta)^{2k\delta}+n \log \Phi(Y))$.
Then we obtain
$L(X,\eps,r_X)$ by Lemma~\ref{lem:lowlevel} in 
$O(nk(\rho^2/\beta)^{3k\delta})$
time. 

So the running time $T(n,k,\rho,\eps)$ of our algorithm satisfies
the relation
\[
	T(n,k,\rho,\eps)=T(n,k_1,\rho,\beta)+ 
	T(n,k_2,\rho,\beta)+
	O(nk(\rho^2/\beta)^{3k\delta}+n \log \Phi(Y)),
\]
where 
$k_1+k_2=k$, $1\leq k_1 \leq k-1$ and $\beta=\eps/(8k-8)$.
This expression expands into a sum of $k-1$ terms 
$O(nk(\rho^2(8k)^k/\eps)^{3k\delta}+n\log \Phi(Y))$ and
$k$ terms $O(n)$ for the base cases. Hence we have
$T(n,k,\rho,\eps)=2^{O(k^2 \log k)} n(\rho^2/\eps)^{3k\delta}+
O(kn\log \Phi(Y))$. We just proved the following:

\begin{theorem}\label{thm:algorithm}
	The $(\rho,\eps)$-distortion problem can be solved in
	$2^{O(k^2 \log k)} (\rho^2/\eps)^{3k\delta}n+
	O(kn\log \Phi(Y))$ time.
\end{theorem}

Our algorithm may also allow us to find more than one matching.
If we return all the matchings stored in $L(X,\eps,r_X)$, we obtain
a collection of $(1+\eps)\rho$-matching that approximates all the
solutions: By Property~\ref{prop:DXY}c,  
for each $\rho$-matching $\sigma$,
one of the  $(1+\eps)\rho$-matching $\sigma_\eps$ that we
return approximates $\sigma$ in the sense that 
$d_M(\sigma,\sigma_\eps) \leq \eps r_X/\rho^2<  2\eps \diam(X)/\rho$.

\section{Improved Algorithm} \label{sec:improved}

In this section, we show how the approximation algorithm 
presented in Section~\ref{sec:decision} 
can be improved using an efficient approximate nearest-neighbor (ANN)
data structure.
More precisely, we will use this data structure to
answer fixed-radius near neighbor queries efficiently.

First we briefly introduce $(3/2)$-approximate nearest neighbor searching.
Let $(S,d_S)$ be a doubling space, with doubling dimension $\delta_S$.
Let $T$ be an $n$-point subset of $S$. 
Given a query point $q \in S$, its nearest neighbor in $T$
is a point $t_q \in T$ such that $d_S(q,t_q)$ is minimum.
A $(3/2)$-approximate nearest neighbor (ANN) is a point $\hat q \in T$
such that $d_S(q,\hat q) \leq (3/2)d_S(q,t_q)$. The data structure
by Cole and Gottlieb~\cite{CG06} 
returns a $(3/2)$-ANN in time $2^{O(\delta_S)}\log n$.
This data structure is dynamic: It allows the insertion or 
deletion of a point in $2^{O(\delta_S)}\log n$ time.

\subsection{Fixed-radius near neighbor searching}\label{sec:fixednn}

We now consider a variation on the problem above. Given 
a radius $R$ and a query point $q$, a 
{\it fixed-radius near neighbor query} returns all the points in $T$ that
are at distance at most $R$ from $q$. We now show
how to efficiently answer these queries in a doubling space.

\begin{lemma}\label{lem:fixednn}
	Suppose that $\dmin(S)\geq R/\lambda$	for some constant $\lambda>0$.	 	
	Then Gottlieb and Cole's ANN data structure~\cite{CG06}  allows us to 
	answer fixed-radius near 	neighbor queries in $2^{O(\delta_S)}\log n$ time. 
	Each such query returns $2^{O(\delta_S)}$ points.
	This data structure also supports insertions and 
	deletions in $2^{O(\delta_S)}\log n$ time.
\end{lemma}

We now show how to answer a query. Let $q$ be a query point, so we
want to return all the points in $T \cap \ball(q,R)$. We repeatedly
perform (3/2)-ANN queries, as long as the distance $d_S(q,\hat q)$ from $q$
to its approximate near neighbor $\hat q$ is at most $3R/2$.
If $d_S(q,\hat q) \leq R$, then we report $\hat q$. In any case,
we delete $\hat q$ from $T$ and insert it into a set $U$.
In the end, we find a point $\hat q$ such that $d_S(q,\hat q) > 3R/2$,
and we complete the procedure by inserting 
all the points of $U$ back into $T$.

As $U \subseteq \ball(q,3R/2)$ we have $\diam(U)\leq 3R$.
Since $U \subseteq S$ we also have $\dmin(S) \geq R/\lambda$.
Therefore, $\Phi(U) \leq 3\lambda$, and since $U$ has doubling
dimension $\delta_S$, we have $|U| =2^{O(\delta_S)}$ by
Lemma~\ref{lem:packing}. Therefore, we answer one fixed-radius
near neighbor query by performing $2^{O(\delta_S)}$ ANN queries,
insertions and deletions. So a fixed-radius near neighbor
query is answered in $2^{O(\delta_S)}\log n$ time.

\subsection{Faster pruning}

The bottleneck of our approximation algorithm is 
when we check whether a matching $\sigma_\eps$ is at distance
at least $\eps r_X/(2\rho^2)$ from each matching 
$\sigma_\eps' \in M(y,X,\eps,r_X)$.
(See Algorithm~\ref{alg:lowlevel}.)
As  $|M(y,X,\eps,r_X)|=O((\rho^2/\eps)^{k\delta})$,
it takes time $O(k(\rho^2/\eps)^{k\delta})$ for each 
matching $\sigma_\eps$ using brute force. 

Instead of checking by brute force, we record each set $M(y,X,\eps,r_X)$ 
in a $(3/2)$-ANN data structure. Remember that we identify the space
of matchings from $X$ to $Y$ with $Y^k$, and thus this space has
doubling dimension $k\delta$ by Corollary~\ref{cor:dimmatching}.
By Property~\ref{prop:DXY}b, we have 
$\dmin(M(y,X,\eps,r_X))\geq\eps r_X/(2\rho^2)$.
By Lemma~\ref{lem:fixednn},
this data structure allows us to  perform insertions 
and fixed-radius near neighbor 
queries in $2^{O(k\delta)}\log n$ time with radius $R=\eps r_X/(2\rho^2)$. 
As we assume that $\delta=O(1)$, this is $2^{O(k)}\log n$ time.
We only need one such query to determine whether 
 $\sigma_\eps$ is at distance
at least $\eps r_X/(2\rho^2)$ from each matching 
$\sigma_\eps' \in M(y,X,\eps,r_X)$.

In summary, we replace the factor $O(k(\rho^2/\eps)^{k\delta})$
in the time bound from Lemma~\ref{lem:lowlevel} with a factor
$2^{O(k)}\log n$.
So the running time of our approximation algorithm is 
given by the recurrence relation (see Theorem~\ref{thm:algorithm}):
\[
	T(n,k,\rho,\eps)=T(n,k_1,\rho,\beta)+ 
	T(n,k_2,\rho,\beta)+
	2^{O(k)}(\rho^2/\beta)^{2k\delta} n \log n+O(n \log \Phi(Y)),
\]
We thus obtain the following time bound.

\begin{lemma}\label{lem:improved}
	The $(\rho,\eps)$-distortion problem can be solved in 
	$2^{O(k^2 \log k)} (\rho^2/\eps)^{2k\delta}n\log n+O(k n\log\Phi(Y))$ time.
\end{lemma}

\subsection{Removing the dependency on the spread}
We now show how to remove the $O(kn\log\Phi(Y))$ term in 
the above time bound. This term comes
from the construction and traversal of the navigating net. But we do not
need to compute more than $k$ levels of the navigating nets.
More precisely, we need to compute it at scale $r_X$, and at
each scale $r_{X'}$ where $X'$ is obtained by recursively 
splitting $X$ into two subsets according to Lemma~\ref{lem:splitting}.
In order to perform our bottom-up construction, we will also
need to be able to find, for two such scales $r<r'$, 
an ancestor in $Y_{r'}$  for any node $y \in Y_r$.
We now explain how to achieve this using the ANN data structure.

We first show how to compute $Y_r$ in $O(n \log n)$ time.
We construct $Y_r$ incrementally: For each point $y \in Y$,
we check whether it is at distance at least $r$ from each previously
inserted point. It can be checked in $O(\log n)$ time by
performing a fixed-radius near neighbor query. 
If $y$ is indeed at distance at least $r$ from each previously
inserted point, we insert $y$ into $Y_r$.
In order to perform this test efficiently, we record $Y_r$
in an ANN data structure. By Lemma~\ref{lem:fixednn},
as $\dmin(Y_r) \geq r$, it
allows us to answer fixed-radius near neighbor queries in
$O(\log n)$ time with radius $R=r$. Insertions are also
performed in $O(\log n)$ time, so we obtain the following:
\begin{lemma}\label{lem:rnet}
	Let $Y$ be an $n$-point metric space of constant doubling dimension.
	We can compute an $r$-net of $Y$ 	in $O(n \log n)$ time.
\end{lemma}

Our algorithm for the $(\rho,\eps)$-distortion problem requires
us to know the horizontal edges of each $r$-net we consider.
They can also be computed in $O(n \log n)$ time using
the same approach as the lemma above, except that the radius
$R$ is now $6r$.

\begin{lemma}\label{lem:computezrh}
	We can compute all the horizontal edges of $Y_r$ in 
	$O(n \log n)$ time.
\end{lemma}

Finally, at scales $r<r'$, we need to be able to find the
ancestor in $Y_{r'}$ for each node $y \in Y_r$, in order to run
our bottom-up construction of $L(W,\beta,r')$ from
$L(W,\beta,r)$. As we no longer have vertical edges, 
we relax the definition of ancestors as follows:
The {\em ancestor} of $y \in Y_r$  is
a point $y' \in Y_{r'}$ such that 
$d_Y(y,y') \leq r'$. (There could be several such points;
 we can pick any one of them to be the ancestor of $y$.)
This property suffices for our algorithm to work as
described in Section~\ref{sec:decision}.

In order to compute the ancestors, we construct an ANN data
structure for $Y_{r'}$.  By Lemma~\ref{lem:fixednn},
as $\dmin(Y_{r'}) \geq r'$, we can construct the ancestor of any
point in $O(\log n)$ time by performing a fixed radius
near-neighbor query with $R=r'$.

\begin{lemma}\label{lem:ancestorz}
	Given $Y_r$ and $Y_{r'}$ such that $r'>r$, we can compute
  $O(n \log n)$ time the ancestors in $Y_{r'}$ of all
	the points in $Y_r$.
\end{lemma}

This shows that we can run the algorithm from 
Section~\ref{sec:decision}  without computing the whole navigating net,
hence the $O(nk\log \Phi(Y))$ term in Lemma~\ref{lem:improved}
is replaced with $O(kn \log n)$, as we only compute $Y_r$ for $k$
different values of $r$. In summary, we obtain the following result.

\begin{theorem}
	\label{th:nospread}
	The $(\rho,\eps)$-distortion problem can be solved in 
	$2^{O(k^2 \log k)} (\rho^2/\eps)^{2k\delta}n\log n$ time.
\end{theorem}

\section{Minimum distortion}\label{sec:mindist}

In this section, we show how to extend our results on the
$\rho$-distortion problem (Sections~\ref{sec:hardness},~\ref{sec:decision}
and~\ref{sec:improved}) to the minimum distortion problem.

\subsection{Hardness result}
We  show how to modify the construction from Section~\ref{sec:hardness} 
so as to obtain a reduction
to the minimum distortion problem. Let $\lambda$ be a large enough number,
say $\lambda=5m\rho^2 2^{k}$.
We add an extra point $x_{k+1}$ to $X$ such that $d_X(x_{i},x_{k+1})=\lambda$ for all $i\leq k$.
We also add a point $p_{k+1}$ to $Y$ that is at distance $\lambda$ from every other
point. If there is a $k$-clique in $G$, then the mapping with minimum distortion
is the same as in Section~\ref{sec:hardness}, where we add 
$\sigma_{x_{k+1}}=y_{k+1}$. It satisfies $\expansion(\sigma)=1$ and 
$\expansion(\sigma^{-1})=\rho$, hence $\dist(X,Y)=\dist(\sigma)=\rho$.

Conversely, suppose that $\dist(X,Y)=\rho$. We first argue that we must have
$\sigma(x_{k+1})=p_{k+1}$. If it were not the case, then 
$d_Y(\sigma(x_1),\sigma(x_{k+1}))\leq 2^k$ so $\expansion(\sigma^{-1})\geq 5m\rho^2$.
As $d_X(x_1,x_2)=4\rho$ and $d_Y(\sigma(x_1),\sigma(x_2)) \geq 1/m$, we also have
$\expansion(\sigma) \geq 1/(4\rho m)$ and thus $\dist(x,y)\geq 5 \rho/4$, a
contradiction. 

Therefore, we must have $\sigma(x_{k+1})=y_{k+1}$. It implies that 
$\expansion(\sigma) \geq 1$. Then the only way to obtain $\dist(\sigma) \leq \rho$
is to use the same construction as in Section~\ref{sec:hardness}, where
we match each point $x_i$ with a point in $R_i$ for $3\leq i \leq k$.
The corresponding vertices in the $k$-clique instance form a clique.

\begin{corollary}\label{corr:minimum}
	Suppose that $X$ and $Y$ have doubling dimension $\log_2 3$.
	For any $\rho \geq 1$, the problem of deciding whether the minimum distortion 
	$\dist(X,Y)$ is equal to $\rho$ 	is NP-hard, and is $W[1]$-hard when  parameterized by $k$.
	It cannot be solved in time $f(k) \cdot n^{o(k)}$ for any computable function
	$f$, unless  ETH is false.
\end{corollary}

\subsection{A first algorithm}
We now present a decision algorithm. Given two positive reals $e$, $e'$,
we would like to know whether there is a mapping $\sigma: X \to Y$ such
that $\expansion(\sigma) \leq e$ and $\expansion(\sigma^{-1}) \leq e'$.
As observed by Kenyon et al.~\cite{Kenyon09}, this is true if and only if there is
mapping $\sigma': X \to Y$ with 
$\expansion(\sigma') \leq \sqrt{ee'}$ and $\expansion(\sigma'^{-1}) \leq \sqrt{ee'}$,
where  we use the metric $d'_{Y}=\sqrt{e'/e}\cdot d_Y$ on $Y$, and we still use
$d_X$ on $X$. In other words, we are solving the $\rho$-distortion problem
with $\rho=\sqrt{ee'}$ on $(X,d_X)$ and $(Y,d'_Y)$. So we obtain
an approximate answer using our algorithm from Theorem~\ref{th:nospread}:
\begin{lemma}\label{lem:decision}
	Given $e,e'>0$ and $0<\eps \leq 1$, our decision algorithm returns 
		a positive answer if there exists a mapping $\sigma:X \to Y$ such that
  $\expansion(\sigma) \leq e$ and $\expansion(\sigma^{-1}) \leq e'$,
  	and returns a negative answer if there is no mapping $\sigma:X \to Y$ 
  	such that $\expansion(\sigma) \leq (1+\eps)e$ and $\expansion(\sigma^{-1}) \leq (1+\eps)e'$.
  Its running time is $2^{O(k^2 \log k)}  (ee'/\eps)^{2k\delta}n\log n.$
\end{lemma}

Still following the approach by Kenyon et al., we can turn this decision
algorithm into an optimization algorithm by observing that $\expansion(\sigma)$
and $\expansion(\sigma^{-1})$ are ratios of distances between points of $X$
and $Y$, and thus there are $O(k^2n^2)$ candidate values for each.
So for any distinct $x,x' \in X$ and $y,y' \in Y$, we run the algorithm from
Lemma~\ref{lem:decision} using $e=d_Y(y,y')/d_X(x,x')$ and $e'=d_X(x,x')/d_Y(y,y')$
and with a relative error ratio $\eps'=\eps/3$. 
More precisely, we consider these pairs $(e,e')$ by increasing value of the product $ee'$, 
and we stop as soon as
we obtain a positive answer. At this point, there must be a matching $\sigma:X \to Y$
such that $\expansion(\sigma) \leq (1+\eps/3)e$ and $\expansion(\sigma^{-1}) \leq (1+\eps/3)e'$,
and since $0<\eps\leq 1$, we have $\dist(X,Y) \leq(1+\eps/3)^2 \leq (1+\eps)ee'$.
On the other hand, we must have $ee' \leq \dist(X,Y)$ as otherwise,
the algorithm would have halted earlier. So taking $\Delta_\eps=ee'$, we obtain the following:

\begin{lemma}\label{lem:optimization}
	Given $0 < \eps \leq 1$, we can compute 
	a  number $\Delta_\eps$ such that $\dist(X,Y) \leq \Delta_\eps \leq (1+\eps)\dist(X,Y)$
	in time 	$2^{O(k^2 \log k)} (\dist(X,Y)/\eps)^{2k\delta}k^4n^5\log n$.
\end{lemma}

\subsection{Improved algorithm} We now describe a faster version of this optimization
algorithm. To this end, we need to introduce {\it well-separated pairs decompositions} 
(WSPD)~\cite{HPbook}. 
Given a metric space $(S,d_S)$ 
and two subsets $A,B \subseteq S$, we say that $A$ and $B$ are 
$(1/\eps)$-{\it separated} if  $\diam(A) \leq \eps d_S(A,B)$ and 
$\diam(B) \leq \eps d_S(A,B)$. As a result,
for any $a,a' \in A$ and $b,b' \in B$, we have
\begin{equation}\label{eq:WSPD}
	d_S(a,b)/(1+2\eps) \leq d_S(a',b') \leq (1+2\eps)d_S(a,b).
\end{equation}
A WSPD of $S$ of size $m$  is a collection of pairs of subsets
$\{A_1,B_1\}$, $\{A_2,B_2\}$,  $\dots,\{A_m,B_m\}$ of $S$ such that 
$A_i$ and $B_i$ are $(1/\eps)$-separated for all $i$, and for any two distinct points
$a,b \in S$, there is a pair $A_i,B_i$ such that $(a,b) \in A_i \times B_i$ or
$(a,b) \in B_i \times A_i$.

Callahan and Kosaraju~\cite{CallahanK95} showed that a WSPD of size $n/\eps^{O(1)}$
can be implicitly computed in $O(n \log n)+n/\eps^{O(1)}$ time when $S$ is a set of $n$ points
in a fixed-dimensional Euclidean space. In particular, we can obtain a pair
of points $(a_i,b_i) \in A_i \times B_i$ for all $i$ within this time bound.
This result was  generalized to metrics of fixed doubling dimension
by Talwar~\cite{Talwar04}, but with size and construction time bounds
larger by a factor $O(\log \Phi(S))$. Still for doubling metrics, Har-Peled and Mendel gave a randomized 
algorithm that constructs  in  $O(n \log n)+n/\eps^{O(1)}$ time a WSPD of size
$n/\eps^{O(1)}$~\cite{HM06}, which can be derandomized
using Cole and Gottlieb's data structure~\cite{CG06}. 
So if we let $\ell_i=d_S(a_i,b_i)$, we obtain the following.
\begin{lemma}\label{lem:WSPD}
	We can compute in time $O(n \log n)+n/\eps^{O(1)}$ a set of lengths $\{\ell_1,\dots,\ell_m\}$
	such that $m=n/\eps^{O(1)}$ and for every pair of distinct points $a,b \in S$, there exists
	$i$ such that $\ell_i/(1+\eps) \leq d_S(a,b) \leq (1+\eps)\ell_i$.
\end{lemma}
We can now present an approximate decision algorithm for the minimum distortion problem.
Given a distortion $\Delta$, it returns a positive answer if $\dist(X,Y) \leq \Delta$,
and it returns a negative answer if $\dist(X,Y) > (1+\eps)\Delta$. We first compute
using Lemma~\ref{lem:WSPD} a set $\{\ell_1,\dots,\ell_m\}$ such
that $m=n/\eps^{O(1)}$ and for all $y,y' \in Y$, $y \neq y'$,  we have
$\ell_i/(1-\eps/6) \leq d_Y(y,y') \leq (1+\eps/6)\ell_i$ for some $i$. Then for each $i \in \{1,\dots,m\}$,
and for each pair of distinct points $x,x' \in X$, 
we run the decision algorithm of Lemma~\ref{lem:decision} with 
expansions $e=(1+\eps/6)\ell_i/d_X(x,x')$ and $e'=(1+\eps/6)\Delta d_X(x,x')/\ell_i$, 
and with relative error ratio
$\eps/6$. If one of the answers to a calls to the algorithm of Lemma~\ref{lem:decision} is positive,
we return a positive answer, and otherwise we return a negative answer.

We now prove that this algorithm is correct. Suppose that $\dist(X,Y) \leq \Delta$.
Then there exists a mapping $\sigma:X \to Y$ such that 
$\expansion(\sigma) \expansion(\sigma^{-1}) \leq \Delta$. We must have 
$\expansion(\sigma)=d_Y(y,y')/d_X(x,x')$ for some $x,x' \in X$ and $y,y' \in Y$.
So there is $i$ such that 
\[
	\frac{1}{1+\eps/6}\cdot\frac{\ell_i}{d_X(x,x')} \leq \expansion(\sigma) \leq 
(1+\eps/6)\frac{\ell_i}{d_X(x,x')}
\] 
and thus
\[
	\expansion(\sigma^{-1}) \leq \frac{\Delta}{\expansion(\sigma)} \leq
	(1+\eps/6) \Delta \frac{d_X(x,x')}{\ell_i}.
\]
So we have $\expansion(\sigma) \leq e$ and $\expansion(\sigma^{-1}) \leq e'$,
and thus the corresponding call to the decision algorithm of Lemma~\ref{lem:decision}
must have returned a positive answer.

Now suppose that our algorithm returns a positive answer. 
So there must be $x,x' \in X$, $i \in \{1,\dots,m\}$ and
a mapping $\sigma: X \to Y$  such that $\expansion(\sigma) \leq (1+\eps/6)e$ and
$\expansion(\sigma^{-1}) \leq (1+\eps/6)e$, where 
$e=(1+\eps/6)\ell_i/d_X(x,x')$ and $e'=(1+\eps/6)\Delta d_X(x,x')/\ell_i$. It
follows that $\dist(X,Y) \leq (1+\eps/6)^2ee'$, 
and thus $\dist(X,Y) \leq (1+\eps/6)^4\Delta$. As $0 < \eps \leq 1$, it implies
that $\dist(X,Y) < (1+\eps)\Delta$.

So we have proved that this decision algorithm is correct. As we run the algorithm
of Lemma~\ref{lem:decision} $k^2n/\eps^{O(1)}$ times, we obtain the following.
\begin{lemma}\label{lem:decision2}
	Given $0<\eps\leq 1<\Delta$, the approximate decision algorithm above returns 
	a positive answer if $\dist(X,Y) \leq \Delta$, and a negative answer if
	$\dist(X,Y) \geq (1+\eps) \Delta$.
	Its running time is $2^{O(k^2 \log k)} (\Delta^{2k\delta}/\eps^{2k\delta+O(1)})n^2\log n.$
\end{lemma}
We can turn this decision algorithm into an approximation algorithm for the minimum
distortion using exponential search. So we run this algorithm with $\Delta=1,2,4,\dots$ until
it returns a positive answer for $\Delta=2^i$. At this point we know that 
$2^{i-1} < \dist(X,Y) < 2^{i+1}$. Then we perform binary search for $\Delta$
in the interval $[2^{i-1},2^{i+1}]$, still using the decision algorithm
of Lemma~\ref{lem:decision2}. After $O(\log(1/\eps))$ steps, it yields a
$(1+\eps)$-approximation of $\dist(X,Y)$. So we obtain the following result.

\begin{theorem}\label{thm:optimization}
	Given $0 < \eps \leq 1$, we can compute 
	a  number $\Delta_\eps$ such that $\dist(X,Y) \leq \Delta_\eps \leq (1+\eps)\dist(X,Y)$
	in time 	$2^{O(k^2 \log k)} (\dist(X,Y)^{2k\delta}/\eps^{2k\delta+O(1)})n^2\log n$.
\end{theorem}

\bibliographystyle{plain}
\bibliography{distortion}

\end{document}